\documentclass[prodmode,acmtopc]{acmsmall}

\usepackage[english]{babel}

\usepackage[utf8]{inputenc}
\usepackage{amsmath}
\usepackage{amssymb}
\usepackage{tabu}
\usepackage{dcolumn}
\newcolumntype{d}[1]{D{,}{.}{#1} }

\usepackage{graphicx}

\usepackage{pifont}
\newtheorem{claim}{Claim}

\def\polylog{\operatorname{polylog}}

\newcommand{\np}{\mbox{\sf NP}}

\newcommand{\hide}[1]{}

\usepackage{xspace}
\newcommand{\MaxNCG}{\textsc{MaxNCG}\xspace}
\newcommand{\SumNCG}{\textsc{SumNCG}\xspace}

\title{Locality-based Network Creation Games}

\author{
Davide Bilò
\affil{Università di Sassari}
Luciano Gualà
\affil{Università di Roma ``Tor Vergata''}
Stefano Leucci
\affil{Università degli Studi dell'Aquila}
Guido Proietti
\affil{Università degli Studi dell'Aquila, \& Istituto di Analisi dei Sistemi ed Informatica (CNR)}
}

\begin{abstract}
Network creation games have been extensively studied, both from economists and computer scientists, due to their versatility in modeling individual-based community formation processes, which in turn are the theoretical counterpart of several economics, social, and computational applications on the Internet.
In their several variants, these games model the tension of a player between her two antagonistic goals: to be as close as possible to the other players, and to activate a cheapest possible set of links.
However, the generally adopted assumption is that players have a \emph{common and complete} information about the ongoing network, which is quite unrealistic in practice.
In this paper, we consider a more compelling scenario in which players have only limited information about the network they are embedded in. More precisely, we explore the
game-theoretic and computational implications of assuming that players have a complete knowledge of the network structure only up to a given radius $k$, which is one of the most qualified \emph{local-knowledge models} used in distributed computing. To this respect, we define a suitable equilibrium concept, and we provide a comprehensive set of upper and lower bounds to the price of anarchy for the entire range of values of $k$, and for the two classic variants of the game, namely those in which a player's cost --- besides the activation cost of the owned links --- depends on the maximum/sum of all the distances to the other nodes in the network, respectively. These bounds are finally assessed through an extensive set of experiments.
\end{abstract}

\category{G.2.2}{Discrete Mathematics}{Graph Theory}[network problems]

\terms{Algorithms, Theory, Economics}

\keywords{Game Theory, Network Creation Games, Price of Anarchy, Local Knowledge}

\begin{document}
\makeatletter
\def\endbottomstuff{\vskip-13pt
\strut
\end@float}
\runningfoot{\relax}
\def\firstfoot{\relax}
\makeatother

\begin{bottomstuff}
\textcopyright ACM, 2016. This is the author's version of the work. It is posted here by permission of ACM for your personal use. Not for redistribution. The definitive version was published in \emph{ACM Transactions on Parallel Computing}, Vol. 3(1), ACM, 6:1-6:26, 2016. http://doi.acm.org/10.1145/2938426

This work was partially supported by the Research Grant PRIN 2010 ``ARS TechnoMedia'', funded by the Italian Ministry of Education, University, and Research.

A preliminary version of this work appeared in the \emph{Proceedings of the 26th ACM Symposium on Parallelism in Algorithms and Architectures (SPAA'14)}, ACM, 277-286, 2014. DOI: http://dx.doi.org/10.1145/2612669.2612680

Author's addresses: D. Bilò, Dipartimento di Scienze Umanistiche e Sociali,
Università di Sassari, Italy; L. Gualà, Dipartimento di Ingegneria dell'Impresa, Università di Roma ``Tor Vergata", Italy; S. Leucci, Dipartimento di Ingegneria e Scienze dell'Informazione e Matematica, Università degli Studi dell'Aquila, Italy; G. Proietti, Dipartimento di Ingegneria e Scienze dell'Informazione e Matematica, Università degli Studi dell'Aquila, Italy, and Istituto di Analisi dei Sistemi ed Informatica, CNR, Rome, Italy.
\end{bottomstuff}

\maketitle

\section{Introduction}
In a \emph{network creation game} (NCG), we are given $n$ players (identified as the nodes of a graph), which attempt to settle an undirected interconnection network. This is realized by letting each player connect herself \emph{directly} to a subset of players, through the activation of the corresponding set of incident links. These links can then be freely used by everyone, and so a player remains connected to non-adjacent players \emph{indirectly}, i.e., by following a \emph{shortest path} in the currently active network. In such a decentralized process, a player has then to strategically balance the sum of two costs: the \emph{building cost}, which is given by the sum of the costs she incurs in activating her links, and the \emph{routing cost}, which is a function of the length of the shortest paths towards the other players.

Due to their generality, it is evident that NCGs can model very different practical situations. Just to mention an example, NCGs are fit to model the decentralized construction of \emph{communication} networks, in which the constituting components (e.g., routers and links) are activated and maintained by different owners, as in the Internet.
In the very first formulation of the game \cite{JW96}, the building and routing costs are defined as follows. Concerning the building cost, each activated link $(i,j)$ has a cost $c_{ij}$ (respectively, $c_{ji}$) for player $i$ (respectively, $j$), and the formation of a link requires the consent of both players involved (since each of the parties pays the corresponding activating cost), while link severance can be done unilaterally. On the other hand, the routing cost is given by a function additively depending on the distances to all the other players.
Later on, Fabrikant \emph{et al.} \cite{FLM03} developed a simplified version of this model, which is also the most popular one in the field of \emph{Algorithmic Game Theory} (AGT), namely that in which the activation of each link has a fixed cost $\alpha > 0$, and this is incurred by the activating player only and without the consent of the adjacent player, while the routing cost is simply the sum of the distances to all the other players (in graph terminology, this is known as the \emph{status} of a node). Besides this simplification, the merit of such a paper was that of emphasizing how the social utility for a (very large) system as a whole is affected by the selfish behavior of the players, which was instead downplayed by the economists, that were more focused on system stability issues. This new perspective inspired then a sequel of papers in the AGT community, as detailed in the following.

\paragraph{Previous work on NCGs}
More formally, the form of a NCG as provided in \cite{FLM03}, which we call \SumNCG, is as follows: we are given a set of $n$ players, say $V$, where the
strategy space of player $u \in V$ is the power set $2^{V \setminus\{u\}}$. Given
a combination of strategies $\sigma=(\sigma_u)_{u \in V}$, let $G(\sigma)$ denote the underlying undirected
graph whose node set is $V$, and whose edge set is
$E(\sigma)=\{(u,v): u \in V \wedge v \in \sigma_u\}$. 
Then, the \emph{cost} incurred by player $u$ in $\sigma$ is
\begin{equation}
\label{sum} C_u(\sigma) = \alpha \cdot |\sigma_u| + \sum_{v \in V}
d_{G(\sigma)}(u,v)
\end{equation}
\noindent where $d_{G(\sigma)}(u,v)$ is the distance between $u$
and $v$ in $G(\sigma)$.
When a player takes an action (i.e., activates a subset of incident edges), she aims to keep this cost as low as possible. Under the assumption of a complete knowledge of $G(\sigma)$, we therefore have that a player $u$ is fully aware that after switching from strategy $\sigma_u$ to strategy $\sigma'_u$, the network will change to $G(\sigma_{-u}, \sigma'_u)$.
Thus, a \emph{Nash Equilibrium}\footnote{In this paper, we only
focus on \emph{pure}-strategy Nash equilibria.} (NE) for the game is a strategy profile $\bar{\sigma}$ such that for every player $u$ and every strategy profile $\sigma_u$, we have that $C_u(\bar{\sigma}) \leq C_u(\bar{\sigma}_{-u}, \sigma_u)$. If we characterize the space of NE in terms of the \emph{Price of Anarchy}
(PoA), i.e., the ratio between the social cost of the costlier NE to the optimal
(centralized) social cost, then it has been shown this is constant for all values of $\alpha$ except for $n^{1-\varepsilon} \leq \alpha < 65 \, n$, for any $\varepsilon \geq 1/\log n$ (see \cite{MMM13,MS10}), while an upper bound of $2^{O(\sqrt{\log n})}$ is known for the remaining values of $\alpha$ \cite{DHM12}. Moreover, very recently, in \cite{GHLS13} it was proven that for all constant non-integral $\alpha \geq 2$, the PoA is bounded by $1+ o(1)$.

A first natural variant of \SumNCG was introduced in \cite{DHM12}, where the \emph{eccentricity} of a player rather than her status was considered as a measure of centrality in the network, and then the player cost function was redefined
as follows:
\begin{equation}
\label{max} C_u(\sigma) = \alpha \cdot |\sigma_u| + \max \{d_{G(\sigma)}(u,v):v \in V\}.
\end{equation}
\noindent This variant, named \MaxNCG, received further
attention in \cite{MS10}, where the authors improved the PoA of
the game on the whole range of values of $\alpha$, obtaining in this case that the PoA is constant for all values of $\alpha$ except for $129 > \alpha = \omega(1/\sqrt{n})$, while for the remaining values of $\alpha$, it is at most $2^{O(\sqrt{\log n})}$.

Besides these two basic models, many variations on the theme have been defined.
We mention those obtained by limiting the modification a player can do on her current strategy
(see \cite{ADH10,L12,MS12}), or by budgeting either the
number of edges a player can activate or her eccentricity (see \cite{LPR08,EFM11,BGP12}), or finally by constraining the set of available links to a host graph (see \cite{BGLP12,DHM09}).
Generally speaking, in all the above models the obtained upper bounds on the PoA are asymptotically worse than those we get in the two basic models. 

\paragraph{Criticisms to the standard model}
Observe that while the general assumption that players have a \emph{common and complete} information about the ongoing network is feasible for small-size instances of the game, this becomes unrealistic for large-size networks. This is rather problematic, since the asymptotic analysis which guides the AGT literature requires instead a growing size of the input. Moreover, quite paradoxically, the full-knowledge assumption is not simplifying at all: it makes computationally unfeasible for a player to select a best-response strategy, or even to check whether she is actually in a NE!

Very recently, the same observation leads Ballester Pla \emph{et al.} to consider in \cite{CPV09} a more compelling scenario for the related class of \emph{network} (or \emph{graphical}) \emph{games}. In a graphical game, players are embedded in a network, and the cost of a player depends on her action and on those of her neighbors, and thus is correlated to the (knowledge of the) entire network. In \cite{CPV09} the authors assume instead that players have a complete knowledge of the network structure up to a given radius $k$, and use this information to make up a belief about the rest of the network. For this model, they provide a closed formula to compute a \emph{Bayes-Nash equilibrium} for the game, and show an interesting relationship with a scenario in which players have a \emph{bounded rationality} (i.e., they take a step by only exploring a subset of the strategy space).

Another work which constrains the available strategies of the players according to a concept of locality is \cite{H13}. In this work the author studies the non-coordinated process of \emph{matching formation} where each player can (potentially) match with any other player having distance at most $k$ in a graph which depends on the current state of the game.

\paragraph{Our new NCG model}
In this paper we concentrate on \MaxNCG and \SumNCG, in this order of presentation and importance, but we deviate from the standard full-knowledge model, and we explore the theoretical implications on the two games induced by the assumption that players have a partial view of the network. More precisely, we consider the players to have only knowledge of their \emph{$k$-neighborhood} (as in \cite{CPV09}), i.e., each player knows $k$ and the entire network up to the nodes at distance at most $k$ from herself. Furthermore, we also assume that players do not even know the size $n$ of the network (in distributed computing terminology, the system is \emph{uniform}).\footnote{According to the spirit of the game, we assume that in our model the players initially sit on a connected network.}

Despite of this partial knowledge of the network structure, the players keep on using the entire network, and so their cost function is still given by Eqs. \eqref{max} and \eqref{sum}, respectively. However, such a cost must now be revised as implicitly incurred by the players --- as a consequence of using the network --- rather than being explicitly known.
On the other hand, consistently with the model, the strategy space of a player is now restricted to selecting a subset of nodes in her $k$-neighborhood.
So, it is in the best interest of a player to reduce an unknown \emph{global} cost, but with the limitation of only knowing and modifying a \emph{local} portion of the network.
This ambitious task must be modeled through a coherent definition of the players' rational behavior, as we explain in the following.
Actually, a player has a partial (defective) view of the network, and thus before taking a step, she has to evaluate whether such a choice is convenient in \emph{every} realizable network which is compatible with her current view. More formally, let $\sigma_u$ be the strategy played by player $u$, and define $\Sigma|_{\sigma_u}$ to be the set of strategy profiles $\sigma = (\sigma_{-u},\sigma_u)$ of the players such that the network $G(\sigma)$ is realizable according to player $u$'s view. Let

\begin{equation}
\label{eq:delta}
\Delta(\sigma_u,\sigma_u') = \max_{\sigma \in \Sigma|_{\sigma_u}} \{C_u( \sigma_{-u},\sigma'_u)-C_u(\sigma)\}
\end{equation}
\noindent

\hide{let ${\cal G}$ denote the set of realizable networks according to $u$'s view, and let $C_u(\sigma,G)$ denote what her cost would be if the actual network was $G \in \cal G$. Then, let
\begin{equation}
\label{eq:delta}
\Delta(\sigma_u,\sigma_u') = \max_{G \in \cal G} \{C_u((\sigma_{-u},\sigma'_u),G)-C_u(\sigma,G)\}
\end{equation}}

\noindent
denote the worst possible cost difference player $u$ would have in switching from $\sigma_u$ to $\sigma'_u$. For our model, we use a suitable equilibrium concept (weaker than NE), that we call \emph{Local Knowledge Equilibrium} (LKE), and which is defined as a strategy profile $\bar{\sigma}$ such that for every player $u$ and every strategy profile $\sigma_u$, we have that $\Delta(\bar{\sigma}_u,\sigma_u) \geq 0$.

As the set of realizable networks $G(\sigma)$ can be infinite, it might appear that a player is not even able to  determine if a strategy is convenient.
However, in Section~\ref{sect:conv_br} we will show that, in contrast with the intuition, this is not the case. In particular, for \MaxNCG we will show that the worst-case scenario for a player is the one in which the network coincides with her view.
Therefore, the player only needs to take into account her view when evaluating a new strategy, ignoring in some sense the portion of the network she cannot see.
This also means that, in our model, a player behaves \emph{exactly} as she would do in the full-knowledge game played on the graph induced by her $k$-neighborhood.
This crucial analogy (which, in contrast, does not hold so tightly for \SumNCG, as we discuss later) also explains why we put first the local-knowledge version of \MaxNCG in our study: it allows us to fairly compare our game with the traditional \MaxNCG, since we are only concerned on how the bounded view will impact on (the dynamics of) the game, given that the behavior of the players will remain the same.
Besides that, we point out another remarkable property of our model:
differently from the standard model, we have that the computational hardness of establishing an improving strategy is now depending not on the size of the entire network, but only on the size of her $k$-neighborhood.
Therefore, although in principle this latter one could be $\Theta(n)$ already for small values of $k$, we believe that in practice the situation may be quite different, as the size of the known network is expected to be constant (or at least very small) compared to $n$.

Regarding \SumNCG, it is instead easy to see that if a player increases the distance of any
vertex $x$ at distance exactly $k$ in her view, then her cost might increase. Indeed, as the rest of the network is unknown to the player, it might be the case that the distance of a large number of non-viewable nodes will increase as well. Nevertheless, for every other strategy, we show that the worst-case network coincides with the player's view.
In conclusion, the above discussion shows that the player can choose and evaluate her strategies as in the classical \SumNCG, but with the exception of the strategies that increase the distances of the nodes at distance $k$, which are forbidden. Thus, there is a dyscrasia between the full- and the local-knowledge version of the game: in this latter one, a player will actually have a more conservative behavior, since an improving strategy w.r.t. her partial view (i.e., ignoring what is outside her view) would not necessarily be an improving strategy w.r.t. the entire network (as it was the case for \MaxNCG)!

\paragraph{Our results}
Having in mind this solution concept and this asymmetric relationship of the two local-knowledge games w.r.t. to their classical counterpart, we characterize the corresponding space of equilibria w.r.t. the social optimum through the study of upper and lower bounds to the PoA. We remark that, as the set of LKEs is broader than the set of NEs, the PoA in our model can only be worse than the PoA in the full-knowledge model.

First, we consider \MaxNCG, and we give three lower bounds to the PoA which are based on different constructions each holding in different ranges of $k$ and $\alpha$. One of these is based on a dense graph, while the other two will have an high social cost due to their large diameter. In this latter case, the difficulty of the construction relies on the fact that, when $\alpha$ is small, we need to guarantee that no player can decrease her cost by buying new edges. We deal with this issue by carefully exploiting the defective views of the players, i.e., we provide a non-trivial construction where every player is not aware that buying a small number of edges would reduce her cost.

We also provide an upper bound to the PoA by considering both the density and the diameter of an equilibrium graph. In order to prove this bound, we take inspiration from techniques successfully used, for example, in \cite{ADH10,DHM12}, which allow to (lower) bound the number of nodes within a certain distance from a player. However, these techniques cannot be directly applied to our model since they require additional work to cope with the concept of locality.

The bounds to the PoA that arise from the various combinations of these results are discussed in detail in Section~\ref{sec:conclusions_max}, and they are essentially tight for many ranges of $\alpha$ and $k$. Here we just outline some prominent implications of our results.
For example, for constant values of $k$ (regardless of $\alpha$) we are able to exhibit stable graphs having diameter $\Omega(n)$. This immediately implies that the PoA is $\Omega(\frac{n}{1+\alpha})$, which is fairly bad. 
However, one might expect the PoA to decrease for large values of $\alpha$. This is not the case, as we can show a tight lower bound of $\Omega(n^\frac{1}{\Theta(k)})$. This is in sharp contrast with the classical full-knowledge version of the game where the PoA is constant as soon as $\alpha \ge 129$.
On the other hand, when $k$ increases, the PoA decreases, although this happens quite slowly. Indeed, even when $k=O(2^{\sqrt{\log n}})$ and $\alpha=O(\log n)$ the PoA is still $\Omega(n^{1-\epsilon})$ for every $\epsilon>0$.
On the bright side, as soon as $k=\Omega(n^\epsilon)$ for any $\epsilon > 0$, we have that, in every LKE, each player has a complete knowledge of the network, and so the PoA coincides with the PoA of the full-knowledge game, hence it is mostly constant.

Then, we consider the sum version of the game, and we show that some of the lower bound schemes used for \MaxNCG can be extended to \SumNCG as well. In particular, a strong lower bound of $\Omega\big(\frac{n}{k}\big)$ to the PoA holds if $k \le c' \cdot \sqrt[3]{\alpha}$ and $\alpha \le n$, for a suitable constant $c'$. Observe that the latter lower bound is at least $\Omega(n^{\frac{2}{3}})$. Moreover, we show that for $\alpha \le n$, the set of LKEs coincides with the set of NEs as soon as $k \ge c \cdot \sqrt{\alpha}$, for a suitable constant $c$. Thus, in this region the PoA is constant, as a consequence of the corresponding result for the full-knowledge version of the game. Unfortunately, we cannot exhibit non-trivial upper bounds for the remaining values of the $(\alpha,k)$-space, and we leave this as a future goal of our research.

\paragraph{Paper organization}
The paper is organized as follows: in Section~\ref{sect:conv_br} we characterize the player's behavior, and we provide some remarks on the complexity of computing a best-response strategy and on the convergence issues of the iterated version of the game. In Sections~\ref{sect:PoA} and ~\ref{sec:sum} we focus on the main results of this paper, namely the study of the PoA for \MaxNCG and \SumNCG, respectively. Then, in Section~\ref{sec:experiments} we provide an extensive set of experiments, which for a twofold reason we restricted to \MaxNCG: on one hand, we have for it a more exhaustive theoretical characterization of the PoA space to compare with, and on the other hand, as we will explain in more detail in the section, for \MaxNCG is computationally feasible to find a best-response strategy of a player for reasonably large values of $n$ and $k$. Finally, Section~\ref{sec:conclusions} concludes the paper and provides few directions of future research.

\section{Preliminary remarks}
\label{sect:conv_br}

We start by showing that, despite the defective knowledge of the network, a player is able to evaluate whether a strategy is convenient in a worst-case scenario. In particular, when a player $u$ changes her strategy from $\sigma_u$ to $\sigma'_u$, she needs to evaluate $\Delta(\sigma_u, \sigma_u')$ by figuring out a realizable network $G(\sigma)$ maximizing \eqref{eq:delta}. In the following propositions we will characterize this worst-case network for both \MaxNCG and \SumNCG.

Let $H$ be the view of $u$ in $G(\sigma)$, i.e. the subgraph of $G(\sigma)$ induced by the $k$-neighborhood of $u$, and let $G \in \mathcal{G}$ be a generic realizable network w.r.t.\ $H$.
Let $G' = G \setminus ( \{u\} \times \sigma_u) \cup ( \{u\} \times \sigma'_u)$ be the network $G$ after the strategy change.
In a similar manner, let $H'=H \setminus ( \{u\} \times \sigma_u) \cup ( \{u\} \times \sigma'_u)$ be the old view of $u$ modified according to the strategy change. Notice that $H'$ might not coincide with the view of $u$ in $G'$.

\begin{proposition}
	\label{prop:maxncg_best_response}
	In \MaxNCG, the worst-case network maximizing \eqref{eq:delta} coincides with $H$.
\end{proposition}
\begin{proof}
	Consider a generic network $G$. In switching from $\sigma_u$ to $\sigma_u'$ the player $u$ is paying an additional cost of:
	\begin{equation}
	\label{eq:cost_difference_max}
		\alpha(|\sigma_u'|-|\sigma_u|) + \max_{v} d_{G'}(u,v) - \max_{v} d_{G}(u,v).
	\end{equation}
	Let $y$ be the vertex maximizing $\max_{v} d_{G'}(u,v)$. If $d_{G}(u,y)<k$ then $y$ belongs to both $H$ and thus to $H'$ as well, therefore the formula \eqref{eq:cost_difference_max} can be upper-bounded by $\alpha(|\sigma_u'|-|\sigma_u|) + \max_{v \in V(H')} d_{H'}(u,v) - \max_{v \in V(G)} d_{G}(u,v)$, which is attained when $G=H$ (hence $G'=H'$) since any graph in $\mathcal{G}$ is a supergraph of $H$.
	Otherwise, let $x$ be the unique vertex in a shortest path $\pi$ from $u$ to $y$ in $G$ such that $d_G(u,x)=k$. Notice that the subpath of $\pi$ between $x$ and $y$ also lies in $G'$, hence $d_{G'}(x,y) \le d_G(x,y)$. We can rewrite \eqref{eq:cost_difference_max} as follows: $\alpha(|\sigma_u'|-|\sigma_u|) +  d_{G'}(u,y) - \max_{v} d_{G}(u,v) \le \alpha(|\sigma_u'|-|\sigma_u|) +  d_{G'}(u,x) + d_{G'}(x,y) - d_{G}(u,y) \le \alpha(|\sigma_u'|-|\sigma_u|) +  d_{H'}(u,x) + d_{G}(x,y) - d_{G}(u,x) - d_G(x,y) = \alpha(|\sigma_u'|-|\sigma_u|) +  d_{H'}(u,x) - k \le \alpha(|\sigma_u'|-|\sigma_u|) +  \max_{v \in V(H')} d_{H'}(u,v) -  \max_{v \in V(H)} d_H(u,v)$.
\end{proof}

Notice that, according to the above discussion, the players do not even need to know the value of $k$ in order to play the game.
Regarding \SumNCG, let us define as $F$ the set a vertices at distance exactly $k$ from $u$ in $H$.

\begin{proposition}
	\label{prop:sumncg_best_response}
	In \SumNCG, every strategy that in\-creas\-es the distance of some vertex of $F$ in $H'$ is not an improving strategy for $u$. For every other strategy of $u$, the worst-case network maximizing \eqref{eq:delta} coincides with $H$.
\end{proposition}
\begin{proof}
	Consider a generic network $G$. When $u$ switches from $\sigma_u$ to $\sigma_u'$ she pays an additional cost of:
	\begin{equation}
		\label{eq:cost_difference_sum}
		\alpha(|\sigma_u'|-|\sigma_u|) + \sum_v d_{G'}(u,v) - \sum_v d_G(u,v).
	\end{equation}
	We first notice that if there exists a vertex $y$ such that $d_G(u,y)=d_H(u,y)=k$ and $d_{H'}(u,y) > k$ then $u$ is not improving in the worst-case scenario. Indeed we can make \eqref{eq:cost_difference_sum} positive by letting $G$ be equal to the graph $H$ where a large number $\eta$ of nodes has been appended to $y$, as \eqref{eq:cost_difference_sum} becomes at least
	$\alpha(|\sigma_u'|-|\sigma_u|) + \sum_{v \in H} d_{G'}(u,v) - \sum_{v \in H} d_G(u,v) + \eta$.
	Therefore, we restrict to strategies $\sigma_u'$ where if $d_H(u,y)=k$ then $d_{H'}(u,y) \le k$.
	Call $P$ the set of vertices $x$ such that $d_G(u,x) < k$, the formula \eqref{eq:cost_difference_sum} becomes:
	$ \alpha(|\sigma_u'|-|\sigma_u|) + \sum_{v \in P} \big( d_{G'}(u,v) - d_G(u,v) \big) +
	 \sum_{v \not\in P} \big( d_{G'}(u,v) - d_G(u,v) \big) \le \alpha(|\sigma_u'|-|\sigma_u|) + \sum_{v \in P} \! \big( d_{H'}(u,v) - d_H(u,v) \big) + \sum_{v \not\in P} \! \big( d_{G'}(u,v) \linebreak - d_G(u,v) \big) \le \alpha(|\sigma_u'|-|\sigma_u|) + \sum_{v \in P} \big( d_{H'}(u,v) - d_H(u,v) \big)$, since for every $v \not\in P$ we have $d_{G'}(u,v) \le d_G(u,v)$.
	This upper bound to \eqref{eq:cost_difference_sum} is attained when $G=H$.
\end{proof}

We now provide some remarks on the complexity of computing a best-response strategy and on convergence issues.

We start by noticing that the \np-hardness reductions which are known in the full-knowledge model for finding a best response in \SumNCG and \MaxNCG can be extended to our games for every $k \ge 2$ and $k\geq 1$, respectively.
Indeed, the full-knowledge versions of \MaxNCG and \SumNCG are \np-hard for $\alpha=2/n$ (see \cite{MS10}) and for every $1<\alpha<2$ (see \cite{FLM03}), respectively. Both reductions are from the {\sc Minimum Dominating Set} problem, that is the problem of finding a minimum cardinality subset $U$ of vertices of an undirected graph $G$ such that every vertex $v \in V(G)$ has a neighbor in $U$ or is in $U$ itself. In both reductions, any best response of a new single player that joins the network $G$ is buying all the edges towards the vertices of a minimum dominating set of $G$. Since in the full-knowledge versions of \MaxNCG and \SumNCG, the best response of a player is independent from the strategy she is actually playing, both games remain \np-hard even if the new player is initially buying all the edges towards all the other players. Turning back to our games, this is equivalent to say that the new player always sees the entire network. Moreover, for $k\geq 2$, the new player is aware to have the full-knowledge of the network as she knows $k$, and does not see vertices at distance 2 from her. Therefore, by Proposition \ref{prop:maxncg_best_response}, we have that \MaxNCG is \np-hard for every $k\geq 1$ and $\alpha=2/n$. Furthermore, by Proposition \ref{prop:sumncg_best_response}, we have that \SumNCG is \np-hard for every $k\geq 2$ and every $1<\alpha<2$.

Concerning the convergence issue, let us consider the iterated version of the game.\footnote{We assume that the players other than being \emph{myopic} are also \emph{oblivious}, namely at each time they only argue about the current view, without taking care of previous views.} A natural question is whether a better- or best-response dynamics always converges to an equilibrium state.
Unfortunately, a negative answer to this question follows from the divergence results presented in \cite{KL13} on the full-knowledge model for \SumNCG and \MaxNCG, since they are based on an instance having (small) constant diameter. This immediately implies the existence of a cycling best-response dynamics for both games as soon as $k \ge c$ for a constant $c$.

\section{Results for \protect\MaxNCG}
\label{sect:PoA}

For the sake of exposition, but also for all the other reasons we discussed in the introduction, we first analyze \MaxNCG, and then we consider \SumNCG. Moreover, for technical convenience, we will assume $\alpha > 1$ although our constructions can also be extended to the case $\alpha \le 1$. We recall that in the full-knowledge version of the game, for $\alpha > 1$ the spanning star is the social optimum and has a cost of $\Theta(\alpha n + n)$.

\subsection{Lower bounds for \MaxNCG}
\label{sec:lower_bounds_max_m3}

We present three lower bounds to the PoA based on three graphs with high social cost which are in equilibrium for different ranges of $\alpha$ and $k$.
The first is a cycle. We have the following.

\begin{lemma}
	If $k\ge1$ and $\alpha \ge k-1$ then $\mbox{PoA}=\Omega(\frac{n}{1+\alpha})$. 
\end{lemma}
\begin{proof}
	Consider a cycle on $n \ge 2k+2$ vertices where each player owns exactly one edge.
	The view of each player $u$ is a path of length $2k$ with $u$ as the center vertex. In order to decrease her eccentricity $u$ has to buy at least one edge. This will decrease the usage cost of $u$ by at most $k-1$ and increase the building cost of $u$ by at least $\alpha$. Then,	$\mbox{PoA} = \Omega\left( \frac{\alpha n + n^2}{\alpha n+ n} \right) = \Omega\left( \frac{n}{1+\alpha} \right)$.
\end{proof}

The next lower bound is based on a dense graph of large girth.

\begin{lemma}
\label{lemma:lb_poa_max_dense}
For each $2 \le k = o(\log n)$ and $\alpha \ge 1$ the PoA is $\Omega(n^{\frac{1}{2k-2}})$. 
\end{lemma}
\begin{proof}
For each even integer $g \ge 6$ and prime power $q$ there exists a $q$-regular graph of girth at least $g$ with $n$ vertices and $\Omega(n^{1+\frac{1}{g-4}})$ edges \cite{LUW95}.

We choose $g=2k+2$ and construct such a graph.
The view of each player $u$ is a tree of height $k$ with $q(q-1)^{i-1}$ vertices on level $i$. Moreover, the player $u$ owns at most $q$ edges.

In order to reduce her usage cost by $i$, player $u$ must buy at least $q(q-1)^i-q$ additional edges. If we choose $q \ge 3$ then the increase in the building cost will exceed the decrease in the usage cost. Hence, we have that the PoA is at least $\Omega(\frac{\alpha n^{1+\frac{1}{g-4}}}{ (1+\alpha) n})= \Omega(n^{\frac{1}{2k-2}})$.

It can be shown that the previous construction holds for $k = o(\log n)$.
\end{proof}

The last lower bound is based on a sparse graph with large diameter.
The construction is non-trivial and it is a generalization of the graph shown in \cite{ADH10}. Although the precise definition is critical, we now give some intuition on how the graph is built.
Roughly speaking, the original graph resembles a $2$-dimensional square grid that was rotated by $45^\circ$ and had the vertices on the opposite sides identified in order to form a toroidal shape.
This graph has several useful properties: it is vertex-transitive and the diameter is about the length of a ``side'' of the grid. Moreover, if the value of $k$ is small, each player $u$ is not aware of the toroidal shape as she only sees a ``square'' subgraph. This subgraph has $4$ vertices at distance $k$ from $u$ whose pairwise distance are $2k$. This fact can be used to show that, actually, this graph is stable for small values of $\alpha$ and $k$, e.g. $\alpha=k=1$. Unfortunately, this is no longer true for larger values of $k$ since, for example, the addition of $4$ edges suffices to reduce the eccentricity of $\Omega(k)$. Moreover, if $\alpha$ is large, a player has convenience in removing an edge as this results in a constant increase in her eccentricity.
To deal with these issues we generalize this construction in three ways. First, we increase the number of dimensions from $2$ to a parameter $d$ so that the graph now resembles a rotated $d$-dimensional cube grid where each face has been identified with the opposite one. For each vertex $u$ we are now able to find $2^d$ other vertices that are at a distance of $k$ from $u$ and whose pairwise distances are at least $2k$.
Second, in order to get a graph with large diameter, we no longer restrict the dimensions to be equal to each other. Intuitively, instead of starting with a $d$-dimensional ``cube'', we start with a $d$-dimensional hyper-rectangle.
Finally, we ``stretch'' the graph by replacing each edge with a path of length $1 \le \ell = \Theta(\alpha)$ between its endpoints. This causes the addition of $\ell-1$ new vertices per edge. We call these new vertices ``non-intersection vertices'' to distinguish them from the already existing ``intersection vertices''. Non-intersection vertices will buy all the links of the graph and we will show that they cannot remove edges as this would result in an increase of at least $\Omega(\ell)$ in their eccentricity.

\begin{figure}[!t]
	\begin{center}
		\includegraphics[scale=0.45]{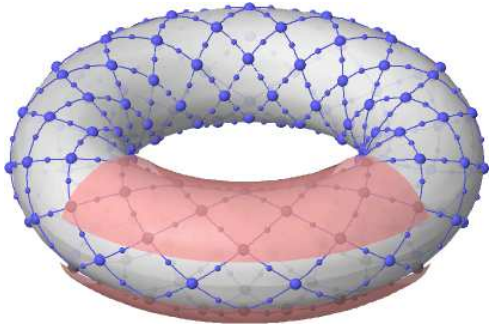}
	\end{center}
	\caption{Graph with $d=2$ dimensions of sizes $\delta_1=15$ and $\delta_2=5$, and $\ell=2$. Intersection vertices have a bigger size than non-intersection vertices. In this example the view of the intersection vertex $(k^{*},k^{*})$ for $k=4$ is in red. Notice that the vertex $(k^{*},k^{*})$ lies on an invisible portion of the torus.}
	\label{fig:3d-torus}
\end{figure}

We now give the details of the construction, which will depend on  $\ell \ge 1$, $d \ge 2$, and $\delta_1,\dots,\delta_d \ge 1$, where $\delta_i$ is the ``length'' of the $i$-th dimension. Vertices will be named after their coordinates and we will interpret the $i$-th coordinate modulo $2 \delta_i$, that is $2\delta_i \equiv 0$.

The graph is built by starting from an empty graph, proceeding in the following way:
add the set of \emph{intersection vertices}, each of these vertices is a $d$-tuple $(\ell \cdot a_1, \ell \cdot a_2, \dots, \ell \cdot a_d)$ such that $a_1 \equiv a_2 \equiv \dots \equiv a_d \pmod{2}$ where each $a_i$ with $1\le i \le d$ is an integer between $0$ and $2\delta_i-1$.
Then, connect each intersection vertex $(x_1, \dots, x_d)$ to $2^d$ other vertices, using a path of length $\ell$ (so if $\ell=1$ we only need to add edges). More precisely, we connect such a vertex to $(x_1 \pm \ell, x_2 \pm \ell, \dots, x_d \pm \ell)$ for every possible choice of the $\pm$ signs. We label the $\ell-1$ non-intersection vertices on the paths by varying the coordinates of the endpoints according to the choice of $\pm$ signs, that is we traverse the path from one endpoint to another and label each non-intersection vertex by adding or subtracting $1$ from the coordinates of the previous vertex. In the following, for convenience, when we choose a vertex, we will assume that the $i$-th coordinate is between $0$ and $2\delta_i-1$.

We will consider graphs where $\delta_1 = \dots = \delta_{d-1} = \left\lceil \frac{k}{\ell} \right\rceil + 1$ and $\delta_d \ge \left\lceil \frac{k}{\ell} \right\rceil + 1$.
Let $k^*=\ell(\delta_1-1)$, $u$ be the vertex $(k^*,\dots,k^*)$, and $u^\prime$ be the vertex $(k^*+\ell, \dots, k^*+\ell)$.

Two examples of this construction, with different parameters, are shown in Figure~\ref{fig:3d-torus} and in Figure~\ref{fig:torus}.

The following result is not hard to prove:
\begin{lemma}
	\label{lemma:torus_distance}
	Let $x=(x_1, \dots, x_d)$ and $y=(y_1, \dots, y_d)$ be two vertices. We have $d(x, y) \ge \max_{1 \le i \le d}\min\{|x_{i}-y_{i}|,2\delta_i\ell-|x_{i}-y_{i}|\}$. If at least one of $x$ and $y$ is an intersection vertex, then the previous inequality is strict.
\end{lemma}
\begin{corollary}
	\label{cor:torus_diameter}
	The diameter of the graph is at least $\ell \cdot \delta_{d}$.
\end{corollary}
\begin{proof}
	By Lemma \ref{lemma:torus_distance} the distance between the vertex $(0, \dots, 0)$ and any vertex whose last coordinate is $\ell \delta_d$ is at least $\ell \delta_d$.
\end{proof}

We also consider an ``open'' version of the previous graph, that is built in a similar way except that we do not treat the coordinates in a modular fashion, so $a_i$ is now between $1$ and $\delta_i$, and we connect intersection vertices (with paths) only when all their coordinates differ by exactly $\ell$. It is not hard to see that the view of each player is isomorphic to a subgraph of this ``open'' graph, and that Lemma \ref{lemma:torus_distance} becomes:
\begin{lemma}
	\label{lemma:open_torus_distance}
	Let $x=(x_1, \dots, x_d)$ and $y=(y_1, \dots, y_d)$ be two vertices in the ``open'' graph.
	We have $d(x, y ) \ge \linebreak \max_{1 \le i \le d} |x_{i}-y_{i}|$,
	If at least one of $x$ and $y$ is an intersection vertex, then the previous inequality is strict.
\end{lemma}

\begin{figure}[!t]
\centering
	\includegraphics[scale=0.7]{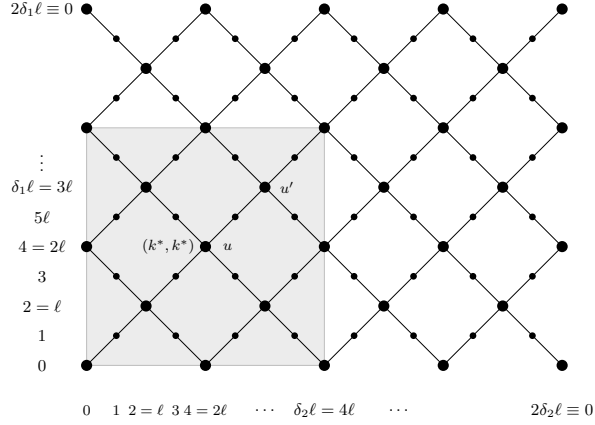}
	\caption{Graph with $d=2$ dimensions of sizes $\delta_1=3$ and $\delta_2=4$, and $\ell=2$. Intersection vertices have a bigger size than non-intersection vertices. The vertices on the first row (resp. column) coincide with the corresponding vertices on last row (resp. column). In this example the view of the intersection vertex $(k^{*},k^{*})$ for $k=4$ is highlighted in gray.}
	\label{fig:torus}
\end{figure}

We now prove a general lemma that will be very useful in the following proofs:
\begin{lemma}
	\label{lemma:shorten_distance}
	Let $H$ be a graph, $u \in V(H)$, and $L = \{v_1, \dots, \linebreak v_{|L|}\} \subseteq V(H) \setminus \{ u \}$.
	If $d_H(u,v_i) \ge h $ for every $1 \le i \le |L|$, $d_H(v_i, v_j) \ge 2h-2$ for every $1\le i,j \le |L|$ with $i \not= j$, and $F$ is a set of edges such that: (i) each edge in $F$ has $u$ as an endpoint, and (ii) $d_{H+F}(u,v_i) < h$ for every $1 \le i \le |L|$, then it holds $|F| \ge |L|$.
\end{lemma}
\begin{proof}
	Every shortest path from $u$ to a vertex $v \in L$ in $H+F$ must use exactly one edge $(u,y) \in F$, and $d_H(v,y) \le h-2$ must hold.
	
	For every other vertex $v^\prime \in L \setminus \{u\}$, the shortest path between $u$ and $v^\prime$ in $H+F$ cannot use the edge $(u,y)$, otherwise we would have:
	\begin{multline*}
		d_{H+F}(u,v^\prime) = 1+d_H(y,v^\prime) = 1+d_H(y,v^\prime)+d_H(y,v) \\ -d_H(y,v)
		\ge d_H(v^\prime, v) - d_H(y,v) \ge 2h-2 - (h-2) \ge h.
	\end{multline*}
	This implies that $F$ must contain at least one edge for each vertex of $L$.
\end{proof}

We now define the ownership of the edges. Consider the path $\langle u=x_0,x_1,\dots, x_{\ell -1},x_\ell=u' \rangle$ from $u$ to $u'$. For $i=1,\dots,\ell-1$, vertex $x_i$ buys the edge towards $x_{i-1}$, and $x_{\ell -1}$ also buys the edge towards $u'$. The ownership of the edges of the other paths are defined symmetrically. Observe that the intersection vertices buy no edges.

Given an intersection vertex $v=(x_1, x_2, \dots, x_d)$ we define $\mathcal{F}^h(v)$ as the set of vertices reachable by traversing an edge incident to $v$, and then proceeding in the same direction for a total of $h$ steps, i.e., $\mathcal{F}^h(v)=\{ (x_1 \pm h, x_2 \pm h, \dots, x_d \pm h)\mbox{, for every possible choice of $\pm$ signs} \}$. If $h \le k$ then $|\mathcal{F}^h(v)|=2^d$ and, by Lemma \ref{lemma:torus_distance}, the distance between $v$ and any of those vertices is exactly $h$.

The following lemmas are instrumental to prove the lower bound to the PoA for \MaxNCG as they will provide sufficient conditions for intersection and non-intersection vertices to be in equilibrium.

\begin{lemma}
	\label{lemma:intersection-equilibrium}
	If $d \ge \log \frac{k-1}{\alpha}$ then the intersection vertices are in equilibrium.
\end{lemma}
\begin{proof}
	By symmetry, let us consider only the intersection vertex $u$.
	As $u$ has not bought any edge, she can modify her strategy only by buying new edges.
	Every vertex in the set $\mathcal{F}^k(u)$ is at distance $k$ from $u$, moreover, by Lemma \ref{lemma:open_torus_distance}, any two distinct vertices in $\mathcal{F}^k(u)$ have a distance of at least $2k$ in the view of $u$.
	
	By Lemma \ref{lemma:shorten_distance}, $u$ needs to buy at least $2^d$ edges in order to reduce her eccentricity.
	If she does so, she saves at most $k-1$ on the usage cost while paying at least $\alpha 2^d$, but we have $\alpha 2^d \ge k-1$.
\end{proof}

\begin{lemma}
	\label{lemma:non-intersection_far_vertices}
	If $k\ge \ell$, for every non-intersection vertex $v$, there is a set $L$ of $2^d$ vertices at distance $k$ from $v$ and at distance at least $2k-\ell$ between each other in the view of $v$.
\end{lemma}
\begin{proof}
	By symmetry, let $v=(k^*+\ell-z, \dots, k^*+\ell-z)$ with $1<z\le \left\lfloor \frac{\ell}{2} \right\rfloor$. The nearest intersecting vertex from $v$ is $u^\prime$.
	Let $L^\prime = \mathcal{F}^{k-z}(u^\prime) \setminus \{ (k^*+\ell-(k-z), \dots, k^*+\ell-(k-z) ) \}$. Any pair of vertices in $L^\prime$ differs by at least one coordinate and, by Lemma \ref{lemma:open_torus_distance}, is at distance at least $2k-2z \ge 2k-\ell$ in the view of $v$.
	
	Let $v^\prime$ be a vertex of $L^\prime$, the shortest path from $v$ to $v^\prime$ must contain either $u$ or $u^\prime$. By Lemma \ref{lemma:open_torus_distance}, if it contains $u$ then we have: $d(v,v^\prime)=d(v,u)+d(u,v^\prime) \ge \ell - z + \ell + k - z = 2\ell +k - 2z \ge k$, otherwise it contains $u^\prime$, and we have: $d(v,v^\prime)=d(v,u^\prime)+d(u^\prime,v^\prime) \ge z+k-z = k$. Therefore $d(u^\prime, v^\prime) \ge k$.
	
	Now consider the vertex $y=(k^*+\ell-z-k, \dots, k^*+\ell-z-k)$. We have that $y \not\in L^\prime$ and, by Lemma \ref{lemma:open_torus_distance}, $d(v, y) \ge k$ and $d(y,v^\prime) \ge 2k$.
	
	The claim follows as we can define $L=L^\prime \cup \{ y \}$.
\end{proof}

\begin{lemma}
	\label{lemma:non-intersection-increase}
	If $\frac{\alpha+1}{2} \le \ell \le \alpha + 2$ and $d \ge\log (\frac{k-1}{\alpha} + 2)$, then every non-intersection vertex $v$ is in equilibrium w.r.t.\ all the strategies that increase the number of bought edges.
\end{lemma}
\begin{proof}
	Notice that each non-intersection vertex has \linebreak bought at most $2$ edges, and consider a strategy that increases the number of bought edges.
	
	If $k \le \ell - 1$ then the building cost increases by at least $\alpha$ while the usage cost decreases at most by $k-1$, but we have $\alpha \ge \ell-2 \ge k-1$.
	
	Otherwise, $k \ge \ell$ and, by Lemma \ref{lemma:non-intersection_far_vertices}, there exists a set of at least $2^d$ vertices at distance at least $k \ge k-\frac{\ell}{2}$ from $v$, and $2k-\ell$ between each other. By Lemma \ref{lemma:shorten_distance}, $v$ needs to have at least $2^d$ incident edges in order to reduce her eccentricity by at least $\frac{\ell}{2}$.
	
	If, in the new strategy, $v$ has less than $2^d$ incident edges then the building cost increases by at least $\alpha$, and the usage cost decreases by at most $\frac{\ell}{2}-1$, but we have $\alpha \ge \frac{\ell}{2}-1$.
	
	If, in the new strategy, $v$ has at least $2^d$ incident edges, then the building cost increases by at least $\alpha(2^d-2)$, and the eccentricity decreases by at most $k-1$, but we have $\alpha(2^d-2) \ge k-1$.
\end{proof}

\begin{lemma}
	\label{lemma:non-intersection-decrease}
	If $\ell \ge \alpha$, every non-intersection vertex $v$ is in equilibrium w.r.t.\ all the strategies that decrease the number of bought edges.
\end{lemma}
\begin{proof}
	By symmetry, let $v=(k^*+z, \dots, k^*+z)$ with $1 \le z < \ell$. The vertex $v$ can decrease the number of bought edges by at most $1$, thus saving $\alpha$ on the building cost.
	
	Let $G^\prime$ be the view of $v$ where the edges incident to the vertex $v$ have been removed, $x=(k^*+z+k, \dots, k^*+z+k)$, and $x^\prime = (k^*+z-k, \dots, k^*+z-k)$.
	If $x$ and $x^\prime$ are not connected in $G^\prime$, then $v$ cannot decrease the number of bought edges.
	
	Otherwise, let $\pi$ be a shortest path between $x$ and $x^\prime$ in $G^\prime$. Let $y=(k^*+z+h, \dots, k^*+z+h)$ be the first vertex in $\pi$ such that the following vertex in $\pi$ is different from $(k^*+i, \dots, k^*+i)$ for all values of $i \in \mathbb{Z}$. Notice that $y$ must be an intersection vertex, and let $y^\prime$ be the first intersection vertex following $y$ in $\pi$. At least one of the coordinates of $y^\prime$ must be $k^*+z+h+\ell$.
	
	We have $d_{G^\prime}(v,x) = d_{G^\prime}(v,x^\prime)=+\infty$ and, by Lemma \ref{lemma:open_torus_distance}:
	\begin{multline*}
		d_{G^\prime}(x,x^\prime) \ge d_{G^\prime}(x,y) + d_{G^\prime}(y, y^\prime) + d_{G^\prime}(y^\prime, x^\prime) \\ \ge k-h + \ell + k + h + \ell \ge 2k +2\ell.
	\end{multline*}
	
	By Lemma \ref{lemma:shorten_distance}, $u$ needs at least $2$ incident edges for her eccentricity to be under $k+\ell$. If $v$ decreases the number of bought edges then $v$ has at most $1$ incident edge, and her usage cost increases by at least $\ell$, but we have $\alpha \le \ell$.
\end{proof}

\begin{lemma}
	\label{lemma:non-intersection-equilibrium}
	If $\alpha \le \ell \le \alpha + 2$, $k \ge \ell$, and $d \ge\log (\frac{k-1}{\alpha} + 2)$, then every non-intersection vertex  $v$ is in equilibrium.
\end{lemma}
\begin{proof}
By Lemma \ref{lemma:non-intersection-increase}	and Lemma \ref{lemma:non-intersection-decrease} $v$ is in equilibrium w.r.t.\ all the strategies that either increase or decrease the number of bought edges.
	
	We now show that $v$ is also in equilibrium w.r.t.\ the strategies that do not change the  number of bought edges. As the building cost of $v$ remains the same, $v$ must save on her usage cost in order to change her strategy.
	
	We will show that $v$ cannot decrease her usage even when she buys the new edges in addition to the ones already bought.

	If $v$ owns only one edge then let, by symmetry, $v=(k^*+z, \dots, k^*+z)$ with $1 \le z < \ell$, and let the bought edge be towards $(k^*+z+1, \dots, k^*+z+1)$. The vertices $(k^*+z+k, \dots, k^*+z+k)$ and $(k^*+z-k, \dots, k^*+z-k)$ are at distance $k$ from $v$ and $2k$ between each other. By Lemma \ref{lemma:shorten_distance}, $v$ needs at least $2$ new incident edges  to decrease her eccentricity under $k$, but she can only add one.

	If $v$ owns two edges then, by symmetry, let $v=(k^* + \ell -1, \dots, k^* + \ell -1)$ so she has an edge towards $u^\prime$.
	
	The vertices in $\mathcal{F}^{k-1}(u^\prime) \setminus { (k+\ell-k+1, \dots, k+\ell-k+1) }$ are at distance $k$ from $v$ and $2k-2$ between each other, in the view of $v$. By Lemma \ref{lemma:shorten_distance}, $v$ needs at least $2^d-1 \ge 3$ new incident edges in order to decrease her eccentricity under $k$, but she can only add two.
\end{proof}

We are now ready to prove the following:

\begin{theorem}
	\label{thm:torus_poa}
	If $1 < \alpha \le k \le 2^{\sqrt{\log n} - 3}$, then the PoA of \MaxNCG is $\Omega\bigg( \frac{n }{ \alpha \cdot 2^{ (\log \frac{k}{\ell} + 3 ) \log \frac{k}{\ell}}}  \bigg)$.
\end{theorem}
\begin{proof}
	Fix $\ell = \left\lceil \alpha \right\rceil$, and notice that $k \ge \ell \ge 2$ holds as $k$ must be an integer.
	Fix $d=\left\lceil \log\left( \frac{k}{\ell} + 2 \right) \right\rceil$, this implies $d \ge 2$. We will use the following inequalities: $\log\left(\frac{k}{\ell}\right) \le d \le \log\left(\frac{k}{\ell}\right) + 3$. Finally, as already said, we set $\delta_1,  \dots, \delta_{d-1}$ to $\left\lceil \frac{k}{\ell} \right\rceil + 1$.

	In order to be $\delta_d \ge \delta_1$ it suffices for $k$ to be at most $2^{\sqrt{\log n} - 3}$, as shown by the following calculations.

	The number of intersection vertices of the graph is $N=2\prod_{i=1}^d \delta_i = 2(\left\lceil \frac{k}{\ell} \right\rceil + 1)^{d-1} \delta_d $ while the total number of vertices is: $n=N+2^{d-1}N(\ell-1)=N(2^{d-1}(\ell-1)+1)$ therefore $N=\frac{n}{2^{d-1}(\ell-1)+1}\ge\frac{n}{2^{d-1}\ell}$.
	\begin{multline*}
		\hfill \delta_d \ge \delta_1 \iff
		\frac{N}{2(\left\lceil \frac{k}{\ell} \right\rceil + 1)^{d-1}} \ge \left\lceil \frac{k}{\ell} \right\rceil + 1 \hfill \\
		 \hfill \Longleftarrow \frac{n}{2^d\ell} \ge \left( \left\lceil \frac{k}{\ell} \right\rceil + 1 \right)^d  \Longleftarrow  k \le  \frac{n^{\frac{1}{d}}}{2} \ell^{1-\frac{1}{d}} - 2\ell \hfill  \\
		 \hfill \Longleftarrow  6 k \le n^\frac{1}{d}  \iff d \log 6k \le \log n \hfill \\
		 \hfill \Longleftarrow (\log k + 3)^2 \le \log n \Longleftarrow k \le 2^{\sqrt{\log n} - 3}. \hfill
	\end{multline*}		
	By Corollary \ref{cor:torus_diameter}, the diameter of the graph is at least:
	\begin{multline*}
	 \ell\,\delta_d = \Omega\bigg( \frac{N \ell^{d}}{2k^{d-1}} \bigg) = \Omega\bigg(\frac{n \ell^{d-1}}{2^d k^{d-1}}\bigg) = \Omega\bigg( \frac{n \ell^{d}}{ k^{d}}  \bigg) \\
	 = \Omega\bigg( \frac{n }{ (\frac{k}{\ell})^{d}}  \bigg) = \Omega\bigg( \frac{n }{ 2^{d \log \frac{k}{\ell}}}  \bigg) = \Omega\bigg( \frac{n }{ 2^{ (\log \frac{k}{\ell} + 3 ) \log \frac{k}{\ell}}}  \bigg).
	\end{multline*}

	By Lemma \ref{lemma:intersection-equilibrium} and Proposition \ref{lemma:non-intersection-equilibrium} the graph is in equilibrium.
	As every vertex in the graph owns at most $2$ edges, the total number of edges is at most $2n$, and the PoA is:
	\[
	\Omega\left( \frac{\alpha n + n \ell \delta_d }{\alpha n} \right) = \Omega\left( \frac{\ell \delta_d }{\alpha} \right) = \Omega\bigg( \frac{n }{ \alpha \cdot 2^{ (\log \frac{k}{\alpha} + 3 ) \log \frac{k}{\alpha}}}  \bigg).  \qed \]
\end{proof}

\subsection{Upper bounds for \MaxNCG}

Given a graph $H$, we denote by $\beta_{H,h}(v)$ the \emph{ball} of radius $h$ centered at node $v$ in $H$, namely the set of vertices whose distance from $v$ in $H$ is at most $h$. When the graph $H$ is clear from the context we will drop the corresponding subscript. The following lemma shows a relation between $k$ and the number of nodes that a player sees in an equilibrium graph $G$. A similar result is shown in \cite{DHM12} for the original game. 

\begin{lemma}\label{lemma:lb_balls}
Let $G$ be an equilibrium graph whose radius is greater than or equal to $k/2$, and let $N=|\beta_{G,k}(u)|$ be the number of nodes that $u$ sees in $G$. If $\alpha \le k-1$ we have that $k=O( \min \{ \sqrt[3]{N \alpha^2}, \alpha \, 4^{\sqrt{\log N}} \})$.
\end{lemma}
\begin{proof}
First, we need to prove that $k= O(\sqrt{N\alpha})$, and do so by showing that $N=\Omega(k^2/\alpha)$. For every $1\leq i\leq k/2$, let $L_i$ be the vertices of $\beta_{G,k}(u)$ whose distance from $u$ is equal to $i$. We show that $|L_i|=\Omega(i/\alpha)$. If $u$ bought the edges towards all the vertices in $L_i$ she would decrease her eccentricity by at least $k-\max\{k-i,i\}=k-k+i=i$ and increase her building cost by $\alpha|L_i|$. As $G$ is an equilibrium graph, we have that $\alpha|L_i|\geq i$, i.e., $|L_i|\geq i/\alpha$. Therefore, $N\geq \sum_{i=1}^{\lfloor k/2\rfloor}|L_i|\geq\sum_{i=1}^{\lfloor k/2\rfloor}i/\alpha=\Omega(k^2/\alpha)$.

Now we prove that $k=O(\sqrt[3]{N \alpha^2})$ by showing that $N=\Omega(k^3/\alpha^2)$. Let $h=\lfloor k/8\rfloor-1$, and let $\bar h=k-2h=k-2\lfloor k/8 \rfloor+2$. By the choice of $h$ and $\bar h$, every path of length less than or equal to $2h$ between two vertices of  $\beta_{G,\bar h}(u)$ is entirely contained in  $\beta_{G,k}(u)$. We select a subset of vertices in $\beta_{G,\bar h}(u)$ as {\em center points} by the following greedy algorithm. First, we unmark all vertices in $\beta_{G,\bar h}(u)$. Then we repeatedly select an unmarked vertex $x$ in $\beta_{G,\bar h}(u)$ as center point, and mark all unmarked vertices in $\beta_{G,\bar h}(u)$ whose distances in the graph induced by $\beta_{G,k}(u)$ are at most $2h$ from $x$.

Suppose that we select $l$ vertices $x_1,x_2,\dots,x_l$ as center points. By construction, every vertex in $\beta_{G,k}(u)$ has distance of at most $4h$ to some center point. If player $u$ bought the $l$ edges towards the $l$ vertices $x_1,x_2,\dots,x_l$, she would decrease her eccentricity w.r.t.\ all the vertices in $\beta_{G,k}(u)$ by at least $k-(4h+1)\geq k-4\lfloor k/8\rfloor +3>k/2$ and increase her building cost by $\alpha l$. Because $G$ is an equilibrium graph, we have $\alpha l\geq k/2$, thus $l\geq k/(2\alpha)$. By the choice of $h$ and $\bar h$, the distance in $G$ between any pair of center points is greater than or equal to $2h+1$, furthermore $\beta_{G,h}(x_i)\subseteq \beta_{G,k}(u)$ for every $i=1,\dots,l$. As a consequence, the balls of radius $h$ centered at the center points are pairwise disjoint, and thus
\begin{equation*}
N=|\beta_{G,k}(u)|\geq \sum_{i=1}^{l}|\beta_{G,h}(x_i)|= l \cdot \Omega(k^2/\alpha)=\Omega(k^3/\alpha^2). 
\end{equation*}
Finally we prove that $k=O\Big(\alpha \,4^{\sqrt{\log N}}\Big)$ by showing that $N=\Omega\Big(2^{\log_4^2(k/\alpha)}\Big)$ for every value of $\alpha \le k-1$.
First, we prove the following useful claim:
\begin{claim}\label{lemma:the_growth_of_balls}
Let $i< k/5$, and let $\bar N=\min_{v \in V}|\beta_{G,i}(v)|$. Either there exists a vertex having eccentricity strictly less than $5i$ or $|\beta_{G,4i+1}(v)|\geq (\bar N i)/\alpha$ for every vertex $v$.
\end{claim}
\begin{proof}
If there is a vertex having eccentricity strictly less than $5i$, then the claim is obvious. Otherwise, for every vertex $v$, we have that the eccentricity of $v$ is greater than or equal to $5i$. Let $S$ be the set of vertices whose distance from $u$ is $3i+1$. By the choice of $i$, every path of length less than or equal to $2i$ between a pair of vertices in $\beta_{G,3i+1}(u)$ is entirely contained in $\beta_{G,k}(u)$. We select a subset of $S$, called {\em center points}, by the following greedy algorithm. First we unmark all vertices in $S$. Then we select an unmarked vertex $x \in S$ as a center point, mark all unmarked vertices in $S$ whose distance from $x$ is less than or equal to $2i$, and assign these vertices to $x$.

Suppose that we select $l$ vertices $x_1,x_2,\dots,x_l$ as center points. We prove that $l\geq i/\alpha$. If player $u$ bought the $l$ edges towards the vertices $x_1,x_2,\dots,x_l$, she would decrease her eccentricity w.r.t.\ all the vertices in $\beta_{G,k}(u)$, by at least $k-\max\{k-i,3i+1\}=k-(k-i)=i$. Because $u$ has not bought these edges, we must have $l\alpha\geq i$.

According to the greedy algorithm, the distance between any pair of center points is greater than or equal to $2i+1$; hence the balls of radius $i$ centered at the vertices $x_j$ are pairwise disjoint. Therefore,
\begin{equation*}
\left |\bigcup_{j=1}^{l}\beta_{G,i}(x_j)\right|=\sum_{j=1}^{l}|\beta_{G,i}(x_j)|\geq l \bar N\geq (\bar N i)\alpha.
\end{equation*}
For every $j=1,\dots, l$, we have $d(u,x_j)=3k+1$, so $\beta_{G,i}(x_j)\subseteq \beta_{G,4i+1}(u)$. Therefore, $|\beta_{G,4i+1}(u)|\geq (\bar N i)/\alpha$.
\end{proof}

Let $\bar N_i=\min_{v \in V}|\beta_{G,i}(v)|$. Because $G$ is a connected equilibrium and $\lceil \alpha \rceil\leq k$, $\bar N_{\lceil\alpha\rceil} \geq \lceil\alpha\rceil$. By Lemma \ref{lemma:the_growth_of_balls}, for every $i<k/5$, either there exists a vertex having eccentricity strictly less than $5i$ or $\bar N_{4i+1}\geq \bar N_{i}$. Define the numbers $a_0,a_1,\dots$ using the recurrence relation $a_i=4a_{i-1}+1$ with $a_0=\lfloor\alpha\rfloor$. By induction, $a_i\geq \alpha 4^i$. If the radius of $G$ is strictly less than $k$, then let $j$ be the least number such that the radius of $G$ is less than or equal to $5a_j$; otherwise, let $j$ be the least number such that $a_j\geq k/5$. As the radius of $G$ is greater than or equal to $k/2$, we have that $a_j=\Theta(k)$. By definition of $j$, $\bar N_{a_{i+1}}\geq (a_i\bar N_{a_i})/\alpha\geq 4^i\bar N_{a_i}$, for every $i<j$. From these inequalities we derive that $\bar N_{a_j}\geq 4^{\sum_{i=0}^{j-1}i}$. But $\bar N_{a_j}\leq N$, so $\sum_{i=1}^{j-1}i=j(j-1)/2\leq \log_4 N$. This inequality implies that $j\leq 1+\sqrt{2\log_4 N}=1+\sqrt{\log N}$. Solving the recurrence relation, $a_j=O(\alpha \, 4^j)=O(\alpha 4^{\sqrt{\log N}})$. As $a_j=\Theta(k)$, $k=O(\alpha 4^{\sqrt{\log N}})$. \qed
\end{proof}

From this, it immediately follows:

\begin{corollary}
\label{cor:complete_knowledge}
If $\alpha \le k-1$ and $k > c \cdot \min \{ n, \sqrt[3]{n \alpha^2}, \linebreak \alpha \, 4^{\sqrt{\log n}}\}$, for a suitable constant $c$, then in every equilibrium graph each player sees the whole graph, thus the set of LKEs coincides with the set of NEs.
\end{corollary}

Now, we provide an upper bound to the diameter of an equilibrium graph.

\begin{lemma}
	\label{lemma:diameter-balls}
	Let $G$ be an equilibrium graph of diameter $d$. If $|\beta_{G,k}(v)| \ge \gamma$ for every vertex $v$, then $ d \le \frac{3kn}{\gamma}$.
\end{lemma}
\begin{proof}
	Let $\pi=(x_0,\dots,x_d)$ be a diametral path of $G$.
	We select a set of vertices $C$ such that the $k$-neighborhoods of the vertices of $C$ are pairwise disjoint and cover all the nodes of $\pi$. 	
	We must have $|C| \gamma \le n$, which implies that $|C| \le \frac{n}{\gamma}$, and thus the diameter of $G$ is at most $|C|(2k+1) = \frac{(2k+1)n}{\gamma} \le \frac{3kn}{\gamma}$.
\end{proof}

\begin{lemma}\label{lemma:UB to diameter}
If $\alpha \le k-1$ the diameter of an equilibrium graph $G$ is $O\Big( \min \{ \frac{n\alpha^2}{k^2},  \frac{k n}{2^{\log_4^2 \frac{k}{\alpha}}} \} \Big)$.
\end{lemma}
\begin{proof}
Consider a generic vertex $v$, and let $N=|\beta_{G,k}(v)|$. From Lemma \ref{lemma:lb_balls}, we have that $k \! = \! O( \min\{\sqrt[3]{N\alpha^2}, \alpha 4^{\sqrt{\log N}} \})$, which implies that $N=\Omega(\max \{ \frac{k^3}{\alpha^2}, 2^{\log_4^2 \frac{k}{\alpha}} \} )$. Now, using Lemma \ref{lemma:diameter-balls}, we have that the diameter of $G$ must be $O\Big( \min \{ \frac{n\alpha^2}{k^2},  \frac{k n}{2^{\log_4^2 \frac{k}{\alpha}}} \} \Big)$. 
\end{proof}

We now derive an upper bound to the density of an equilibrium graph. We argue on the girth of the graph in a way similar to \cite{DHM12}.

\begin{lemma}\label{lemma:density}
The number of the edges of an equilibrium graph $G$ is $O(n^{1+\frac{2}{\min\{\alpha,2k\}}})$.
\end{lemma}
\begin{proof}
Let $g$ be the girth of $G$. We first show that $g \ge 2+ \min\{\alpha, 2k\}$, and then the claim follows from the fact that a graph with girth $g'$ must have at most $O(n^{1+\frac{2}{g'-2}})$ edges \cite{Boll78}. Assume by contradiction that there is a cycle $C$ of length strictly less than $2+\min\{\alpha, 2k\}$. Then consider a player $u$ that owns an edge of the cycle. Since $u$ can see the cycle, she can remove the edge. The deletion would increase the distance to any other node by at most $|C|-2$ while $u$ would save $\alpha > |C|-2$, and hence $G$ cannot be an equilibrium.
\end{proof}

We can now prove the following:
\begin{theorem}
\label{thm:UB_PoA_max}
The PoA of \MaxNCG is
$O\big(n^{\frac{2}{\min\{\alpha,2k\}}} + \frac{n}{1+\alpha} \big)$ if $\alpha \ge k-1$ and
$O\big( n^{\frac{2}{\alpha}} + \min \{ \frac{n\alpha}{k^2}, \frac{nk}{\alpha \cdot 2^{\frac{1}{4}\log^2 \frac{k}{\alpha}}} \} \big)$
if $\alpha \le k-1$.
\end{theorem}
\begin{proof}
If $\alpha \ge k-1$ then the claim follows from Lemma~\ref{lemma:density} and from the fact that the diameter of an equilibrium graph is at most $n-1$.
Otherwise, $\alpha \le k-1$ and the claim immediately follows from Lemma~\ref{lemma:density} and Lemma \ref{lemma:UB to diameter}.
\end{proof}

\subsection{Putting all together}
\label{sec:conclusions_max}

Here we summarize our lower and upper bounds to the PoA for \MaxNCG by showing how they combine depending on the values of $\alpha$ and $k$.

First, recall that whenever the view of the players is sufficiently large, then in every LKE, players actually have a \emph{full knowledge} of the network, and so LKEs coincides with NEs (hence the PoA is the same as in the full knowledge version of the game) as shown in Corollary~\ref{cor:complete_knowledge}. The corresponding region is shown in gray in Figure~\ref{fig:results-max}.
Concerning our three lower bounds, the first one of $\Omega\big(\frac{n}{1+\alpha}\big)$ holds for $\alpha \ge k-1$, i.e. in the regions numbered \ding{173}, \ding{174}, and \ding{177} in Figure \ref{fig:results-max}.
For $1 <  \alpha \le k$ and $k \le 2^{\sqrt{\log n}-3}$ we provided a strong lower bound of $\Omega( \frac{n}{\alpha 2^{\Theta(\log^2\frac{k}{\alpha})}} )$ (regions \ding{172}, \ding{175}, \ding{176} in Figure \ref{fig:results-max}). Notice that when $k=\Theta(\alpha)$ this lower bound boils down to $\Omega(\frac{n}{\alpha})$, which is tight. 
Unfortunately, if $\alpha > k$, the previous lower bound is no longer valid, instead we provided a third lower bound of $\Omega( n^\frac{1}{\Theta(k)} )$ holding for $k=o(\log n)$ (regions \ding{172}, \ding{173}, \ding{174} in Figure~\ref{fig:results-max}).

\begin{figure}[!tb]
\centering
\includegraphics[scale=1]{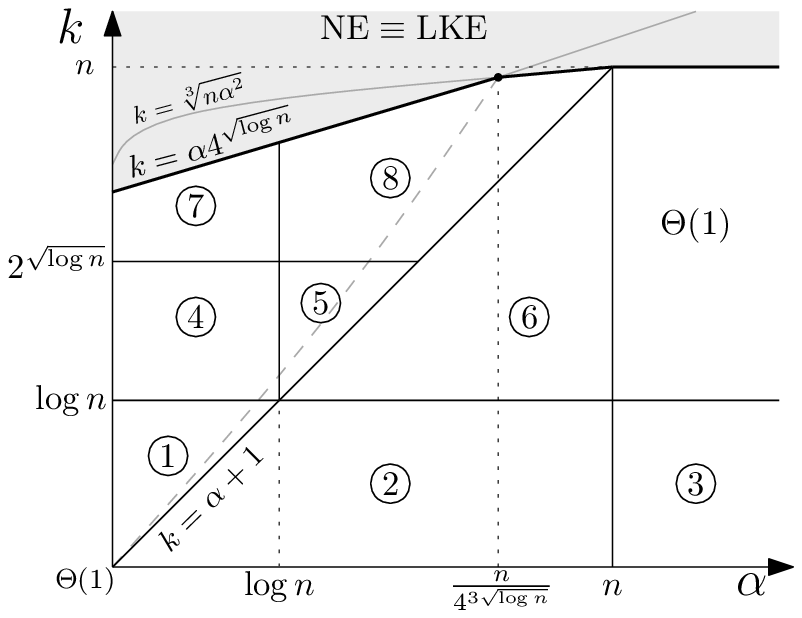}
{\tiny$\mbox{}$\\}
\scalebox{1}{
\begin{tabular}{|c|c|c|}
\hline
\# & Lower Bound & Upper Bound \\ \hline
{\large\ding{172}}  & $\!\!\Omega\bigg( \max\{ \frac{n}{\alpha 2^{\Theta(\log^2 \frac{k}{\alpha})}}, n^\frac{1}{\Theta( k)} \}\bigg)$ & $O\bigg( n^{\frac{2}{\alpha}} + \min \{\frac{n\alpha}{k^2},  \frac{nk}{\alpha 2^{\Theta(\log^2 \frac{k}{\alpha})}} \} \bigg) \!\!$ \\ \hline
{\large\ding{173}} & \multicolumn{2}{c|}{ $\Theta\bigg(\max \{\frac{n}{1+\alpha},n^{\frac{1}{\Theta(k)}} \} \bigg)$ } \\ \hline
{\large\ding{174}}  & \multicolumn{2}{c|}{  $\Theta\bigg(n^{\frac{1}{\Theta(k)}} \bigg)$ } \\ \hline
{\large\ding{175}}  & $\Omega\bigg(\frac{n}{\alpha 2^{\Theta(\log^2 \frac{k}{\alpha})}}\bigg)$ & $O\bigg( n^{\frac{2}{\alpha}} + \min \{\frac{n\alpha}{k^2},  \frac{nk}{\alpha 2^{\Theta(\log^2 \frac{k}{\alpha})}} \} \bigg) \!\!$ \\ \hline
{\large\ding{176}} & $\Omega\bigg(\frac{n}{\alpha 2^{\Theta(\log^2 \frac{k}{\alpha})}}\bigg)$ & $O\bigg( \min \{\frac{n\alpha}{k^2},  \frac{nk}{\alpha 2^{\Theta(\log^2 \frac{k}{\alpha})}} \} \bigg)$ \\ \hline
{\large\ding{177}} & \multicolumn{2}{c|}{ $\Theta\bigg( \frac{n}{1+\alpha} \bigg)$ } \\ \hline
{\large\ding{178}} & \multicolumn{2}{c|}{ $O\bigg( n^{\frac{2}{\alpha}} + \frac{nk}{\alpha 2^{\Theta(\log^2 \frac{k}{\alpha})}} \bigg)$ } \\ \hline
{\large\ding{179}} & \multicolumn{2}{c|}{ $O\bigg( \min\{ \frac{n\alpha}{k^2}, \frac{nk}{\alpha 2^{\Theta(\log^2 \frac{k}{\alpha})}} \} \bigg)$ } \\ \hline
\end{tabular}}

\caption{Lower and upper bounds to the PoA for \MaxNCG. The partition in regions comes up by the combination of the various bounds that we give in Section \ref{sect:PoA}. Notice that the gray region is the set of pairs $(\alpha, k)$ such that, in every LKE, players actually have a full knowledge of the network, and so LKEs coincides with NEs.}
\label{fig:results-max}
\end{figure}

Turning to the upper bounds to the PoA, in Theorem~\ref{thm:UB_PoA_max}, we proved them by considering both the density and the diameter of an equilibrium graph. We showed that the number of edges can be at most $O(n^{1+\frac{2}{\min\{\alpha,2k\}}})$. Regarding the diameter, since for $\alpha \ge k-1$ it can be shown to be $\Omega(n)$, we considered the case $\alpha \le k-1$ where we gave an upper bound of  $O(\min\{ \frac{n \alpha^2}{k^2}, \frac{nk}{2^{\Theta(\log^2 \frac{k}{\alpha})}}) \}$.
Notice that this upper bound is a minimum of two terms. Intuitively, the first one is better when $k$ is not too big w.r.t.\ $\alpha$, e.g.\ when $k=O(\alpha \polylog(\alpha))$. The corresponding region lies between the dashed gray curve and the line of equation $\alpha=k-1$ shown in Figure~\ref{fig:results-max}.

The bounds to the PoA that arise from the various combinations of these results are summarized in Figure~\ref{fig:results-max}. Notice that the bounds for the regions under the line $k=\alpha+1$ are essentially tight.

\section{Results for \protect\SumNCG}
\label{sec:sum}

In this section we provide our results for \SumNCG.  Recall that in the full-knowledge version of the game the spanning star is the social optimum and has a cost of $\Theta(\alpha n + n^2)$. We start with a quite strong lower bound to the PoA.

Consider a graph similar to the one shown in Figure \ref{fig:torus}, whose construction has been described in Section \ref{sec:lower_bounds_max_m3}. In particular, we build such a graph using the following parameters: $d=2$, $\ell=2$, $\delta_1= \left\lceil \frac{k}{2} \right\rceil +1$ and $\delta_2 \ge \delta_1$ which will be specified later.
	We know that $N=2 \delta_1 \delta_2$ and that $n=6 \delta_1 \delta_2$.
	We want $\delta_2 \ge \delta_1 \iff \frac{n}{6\delta_1} \ge \delta_1 \iff \delta_1 \le \sqrt{\frac{n}{6}} \Longleftarrow k \le \sqrt{\frac{2n}{3}} - 4$.
	We have the following:
	
\begin{lemma}\label{lemma:LB_sum_m3}
For $\alpha \ge 4k^3$, the previous graph is an equilibrium.
\end{lemma}
\begin{proof}
	First notice that each vertex $v$ sees at most $4k^2$ other vertices. As $\alpha \ge 4k^3$ and $v$ can reduce the distance to each vertex by at most $(k-1)$, $v$ cannot increase the number of bought edges.
	
	Suppose that the player $v$ changes her strategy from $\sigma_v$ to $\sigma^\prime_v$, let $L$ be the set of vertices that are at distance $k$ from $v$. Consider the view of $v$ before the strategy changes, where the edges of $\sigma_v$ have been removed and replaced with the edges of $\sigma^\prime_v$. Each vertex of $L$ must be at distance at most $k$ from $v$ in this new graph.
	Otherwise, suppose the existence of a vertex $x \in L$ that is at distance at least $k+1$ in this new graph. When $v$ computes $\Delta(\sigma_v, \sigma^\prime_v)$ it will also consider the case where a certain number of (new) vertices $\eta$ are adjacent to $x$, therefore her usage will be at least $k \eta$ which, for a suitable value of $\eta$, is greater than the cost of $v$ in $\sigma$.
	
	Let $v=(k^*+1, k^*+1)$, and consider the set of vertices $L= \left( \mathcal{F}^{k-1}(u) \cup \mathcal{F}^{k-1}(u^\prime) \right)  \setminus \{ (k^*+3-k, k^*+3-k), (k^*+k-1,k^*+k-1) \}$.

	It is easy to see that all the vertices in $L$ are at distance $k$ from $v$, and that
	every vertex $x$ in the view of $v$, that is not a neighbor of $v$, has at most $2$ vertices of $L$ that are at distance at most $k-1$. Finally, every vertex $x$ in the view of $v$ has at least one vertex of $L$ at a distance at least $k$.
	
	This suffices to conclude that every vertex $v$ is currently playing a best response and, therefore, the graph is in equilibrium.
\end{proof}

Using the above lemma we can prove the following:
\begin{theorem}
	Let $\alpha \ge 4 k^3$. For any $k \le \sqrt{\frac{2n}{3}} - 4$, the PoA of \SumNCG is $\Omega(n/k)$ if $\alpha \le n$, and $\Omega(1+\frac{n^2}{k \alpha})$ otherwise.
\label{th:LB_PoA_sum_m3}
\end{theorem}
\begin{proof}
	By the above lemma, the graph is in equilibrium and has diameter $\Omega(\delta_2)=\Omega(\frac{n}{k})$. Moreover, it is easy to see that each vertex has $\Omega(n)$ vertices at distance $\Omega(\frac{n}{k})$. Since the graph has $\Theta(n)$ edges, we have that the cost of the graph is $\Omega(\alpha n +n^2/k)$ while the social optimum is a star with cost $O(\alpha n+n^2)$. The claim follows.
\end{proof}

The following theorem provides a lower bound for a different range of values of the parameters $\alpha$ and $k$.

\begin{theorem}
	\label{th:LBbis_PoA_sum_m3}
	If $\alpha \ge kn$ and $k \ge 2$, the PoA of \SumNCG is $\Omega(n^{\frac{1}{2k-2}})$.
\end{theorem}
\begin{proof}
	We use the same construction of Lemma \ref{lemma:lb_poa_max_dense}.
	Remind that the view of each vertex $v$ is a tree of height $k$ with $q(q-1)^{i-1}$ vertices on level $i$. Therefore $v$ has to buy at least $q$ edges. Moreover, if she buys exactly $q$ edges, then she cannot improve her cost as her neighbors are the medians of the corresponding subtrees.

	As $\alpha \ge kn$, $v$ cannot improve her cost by increasing the number of bought edges. The claim follows.
\end{proof}

Finally, we have:
\begin{theorem}
If $k > 1 + 2\sqrt{\alpha}$, then in every equilibrium graph each player sees the whole graph, thus the set of LKEs coincides with the set of NEs.
\end{theorem}
\begin{proof}
	Let $G$ be an equilibrium with diameter at least $k$ (otherwise the claim is trivially true), and let $u$ and $v$ be two vertices such that $d_G(u,v)=k$.
	By buying the edge $(u,v)$, the player $u$ could decrease the cost needed to reach the last $\left\lfloor \frac{k}{2} \right\rfloor$ vertices along the shortest path between $u$ and $v$ by at least $\frac{k^2}{4}-\frac{2}{4}$.
	As $G$ is an equilibrium we must have $\frac{k^2}{4}-\frac{2}{4} \le \alpha$, which implies $k \le 2\sqrt{\alpha} + 1$.
\end{proof}

The previous results are summarized in Figure~\ref{fig:results-sum}. Notice that it remains open to establish a lower bound to the PoA for values of $k$ between $\Theta(\sqrt[3]{\alpha})$ and $\Theta(\sqrt{\alpha})$. Moreover, we plan in the future to study the corresponding upper bounds.

\begin{figure}[!tb]
\centering
\includegraphics[scale=1]{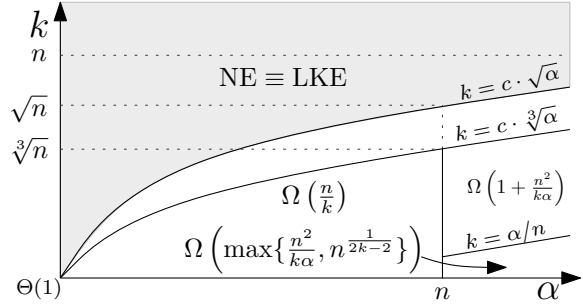}
\caption{Lower  bounds to the PoA for \SumNCG. The behavior of the PoA in the region between the two curves is open. Notice that the gray region is the set of pairs $(\alpha, k)$ such that, in every LKE, players actually have a full knowledge of the network, and so LKEs coincide with NEs.}
\label{fig:results-sum}
\end{figure}

\section{Experimental evaluation}
\label{sec:experiments}

In this section, we present an experimental study whose aim is to complement our theoretical analysis by providing insights on several features of the game.
Indeed, our asymptotic bounds depends on heterogeneous parameters, and so an empirical evidence of their tightness is required. In particular, our study concerns the behavior of best-response dynamics in \MaxNCG, by focusing on the convergence time (i.e., the number of steps needed to form a stable network) and on the structural properties of the resulting equilibria. This way, on one hand our goal is to check whether the theoretical existence of  better- and best-responce cycles is actually frequent, and on the other hand we essentially aim to verify whether stable networks have a social cost which follows the trend suggested by the analysis of the PoA we have provided.
As we will show, the answer to the first question is in the negative, while as far as the second aspect is concerned, we will provide an evidence that our analysis is actually quite accurate (at least when $\alpha$ is not too small).

\subsection{Experiments}

In our experiments, we simulated the best-response dynamics of the players for the \MaxNCG model.
The initial configuration consists of a network randomly sampled from a certain class of graphs, as we will describe in the following.
The players play in turns, following a round-robin policy, i.e., in each round we consider all the players, one at time. When a player is considered, we compute a best-response strategy according to her local knowledge of the network, and whenever this strategy is strictly better than the current one we update the network. Then, we move to the next player. We continue this process until we attain an equilibrium, i.e., we find a round where no player is able to improve her cost.
Moreover, since the convergence is not guaranteed, whenever the execution time exceeds a given threshold, we check if the last strategy profile of the current round already appeared as the last strategy profile of any previous round. In this case, since of the round-robin policy, we can conclude that the best-response dynamics admits a cycle, hence we stop the simulation as we know that no equilibrium will ever be reached by the players.

For each starting network, the above experiment is repeated with different values of the parameters $\alpha$ and $k$ taken from the set $\{ 0.025, 0.05, 0.1, 0.2, 0.3, 0.4, 0.5, 0.7, 1, 1.5, 2, 3, 5, 7, 10 \}$ and from the set $\{ 2, 3, 4, 5, 6, 7, 10, 15, 20, 25, 30, 1000 \}$, respectively. The case $k=1000$ corresponds  to the classical version of the game where every player has no restriction on her view.

In order to compensate for the variability in measurements, for each pair of the parameters $\alpha$ and $k$ we ran $20$ distinct experiments starting from different random graphs belonging to the same class. Overall, we simulated about $36\,000$ different dynamics and, after each round, we collected several different features of the current network such as: diameter, social cost, maximum/average degree, minimum/maximum/average number of bought edges, minimum/maximum/average number of vertices in the view of the players, along with others. Moreover, for each run we also collected global statistics as the number of rounds needed to reach an equilibrium (if any), and the total number of strategy changes performed by the players.

In Figures~\ref{fig:viewsize}-\ref{fig:convergence}, we will show how the mean values of those features (taken over the $20$ different starting networks) vary as a function of either $\alpha$, $k$ or $n$. We also report the corresponding $95\%$ confidence intervals, as long as they do not affect readability.

\subsection{Input classes of graphs}

We considered the following classes of graphs:

\begin{itemize}
	\item \emph{Random trees}: for a given number $n$ of vertices, we picked a tree uniformly at random from the set of all possible trees on $n$ vertices. We determined the ownership of each edge $(u,v)$ of the tree by choosing between $u$ and $v$ with a fair coin toss. The above was repeated for all the values of $n$ in  $\{ 20, 30, 50, 70, 100, 200 \}$. In Table~\ref{table:tree}, we report some statistics of the generated trees for the different values of $n$.

\begin{table}
	\centering
	\tbl{Statistics for the random trees used in the experimental evaluation.
	In each row, $20$ random trees with the same number $n$ of nodes are considered. The value of $n$ is reported in the first column. The remaining columns contain the average statistics over the corresponding trees along with their $95\%$ confidence intervals. 
	}{
	\begin{tabular}{|d{0}|d{2}|d{2}|d{2}|}
		\hline
		n & \multicolumn{1}{c|}{\mbox{Diameter}} & \multicolumn{1}{c|}{\mbox{Max. degree}} & \multicolumn{1}{c|}{\mbox{Max. Bought Edges}} \\ \hline
		20  & 10.65 \pm 0,76 & 4.00 \pm 0,26 & 2.75 \pm 0,34 \\ \hline
		30  & 13.90 \pm 1,18 & 4.30 \pm 0,31 & 3.15 \pm 0,31 \\ \hline
		50  & 19.55 \pm 1,48 & 4.60 \pm 0,35 & 3.30 \pm 0,31 \\ \hline
		70  & 22.10 \pm 1,57 & 5.05 \pm 0,39 & 3.50 \pm 0,32 \\ \hline
		100 & 25.15 \pm 1,95 & 5.35 \pm 0,35 & 3.45 \pm 0,28 \\ \hline
		200 & 43.20 \pm 3,95 & 5.30 \pm 0,31 & 3.85 \pm 0,31 \\ \hline	
	\end{tabular}}
	\label{table:tree}
\end{table}
	
	\item \emph{Erd\H{o}s-Rényi random graphs}: we generated random graphs according to the classical $G(n,p)$ model \cite{ER59} in which a graph of $n$ vertices is built by adding independently each (undirected) edge with a probability of $p$. The parameters $n$ and $p$ were chosen so that the resulting graph was likely to be connected. Any remaining unconnected graph was discarded and regenerated from scratch.  We generated graphs with both $100$ and $200$ vertices and, for each $n$, we set the values of $p$ in order to produce three graphs with different densities. As for trees, the owner of each edge was chosen uniformly at random between its endpoints. A summary of the resulting graphs can be found in Table~\ref{table:erdos}.

\begin{table}
	\centering
	\tbl{Statistics for the different Erd\H{o}s-Rényi random graphs used in the experimental evaluation.
	In each row, a different set of the parameters $n$ and $p$ is considered (first two columns). For each of them, $20$ different graphs have been randomly sampled. The remaining columns report the average statistics over these graphs along with their $95\%$ confidence intervals.}{ 
	\begin{tabular}{|d{0}|d{3}|d{2}|d{2}|d{2}|d{2}|}
		\hline
		n & p & \multicolumn{1}{c|}{\mbox{Edges}} & \multicolumn{1}{c|}{\mbox{Diameter}} & \multicolumn{1}{c|}{\mbox{Max. degree}} & \multicolumn{1}{c|}{\mbox{Max. Bought Edges}} \\ \hline
		100 & 0,060  & 301.10  \pm \ \ 7,51  & 5.30 \pm 0,22  & 12.50 \pm 0,67  & 7.90  \pm 0,43  \\ \hline
		100 & 0,100  & 494.40  \pm 10,86 & 4.00 \pm 0,00  & 18.45 \pm 0,82  & 11.40 \pm 0,69  \\ \hline
		100 & 0,200  & 984.35  \pm 12,23 & 3.00 \pm 0,00  & 29.90 \pm 1,26  & 18.25 \pm 0,57  \\ \hline
		200 & 0,035  & 698.25  \pm 11,22 & 5.25 \pm 0,21  & 15.35 \pm 0,61  & 9.20  \pm 0,42  \\ \hline
		200 & 0,050  & 992.45  \pm 10,18 & 4.20 \pm 0,19  & 19.30 \pm 0,55  & 12.50 \pm 0,65  \\ \hline
		200 & 0,100  & 2005.55 \pm 12,87 & 3.00 \pm 0,00 & 32.80  \pm 1,11  & 18.95 \pm 0,54  \\ \hline
	\end{tabular}}
	\label{table:erdos}
\end{table}
	
\end{itemize}

\subsection{Computing a best-response strategy}

As we pointed out in the theoretical part, the problem of computing a best response for a player is \np-hard. In order to deal with nontrivial values of $n$, we reduced the problem of computing a best response to the problem of computing (a variant of) the minimum dominating set of a suitable graph, and then we used the Gurobi solver \cite{gurobi} on classical integer linear programming formulation of minimum dominating set.
We now sketch the idea behind our approach.

Let $u$ be the player for which we want to compute a best-response strategy. First of all, from Proposition~\ref{prop:maxncg_best_response} we can just compute the best move of $u$ with respect to her view $H$. Now, we remove $u$ from the graph $H$, and we guess the eccentricity $h$ of $u$ in the graph resulting from a best-response move. If $H'$ denotes the $(h-1)$-th power\footnote{The $h$-th power of an undirected graph $G$ is a graph on the same vertex-set of $G$ where the edge $(u,v)$ exists iff the distance between $u$ and $v$ in $G$ is at most $h$.} of the graph $H \setminus \{u\}$, then it is easy to see that the original problem is equivalent to that of finding a minimum dominating set of $H'$ in which certain vertices are constrained to be included in the solution, namely those who bought an edge towards $u$ in the current strategy profile.

\subsection{Experimental results}

Here we discuss some interesting features of the experimental results concerning \MaxNCG. 

\paragraph{Knowledge of the network}

We start by arguing on how the size of the network known by players at equilibrium varies as a function of $k$.
This is both interesting by itself, and it will also be instrumental for further discussions regarding the quality of equilibria.

\begin{figure}[t]
	\includegraphics[width=0.5\textwidth]{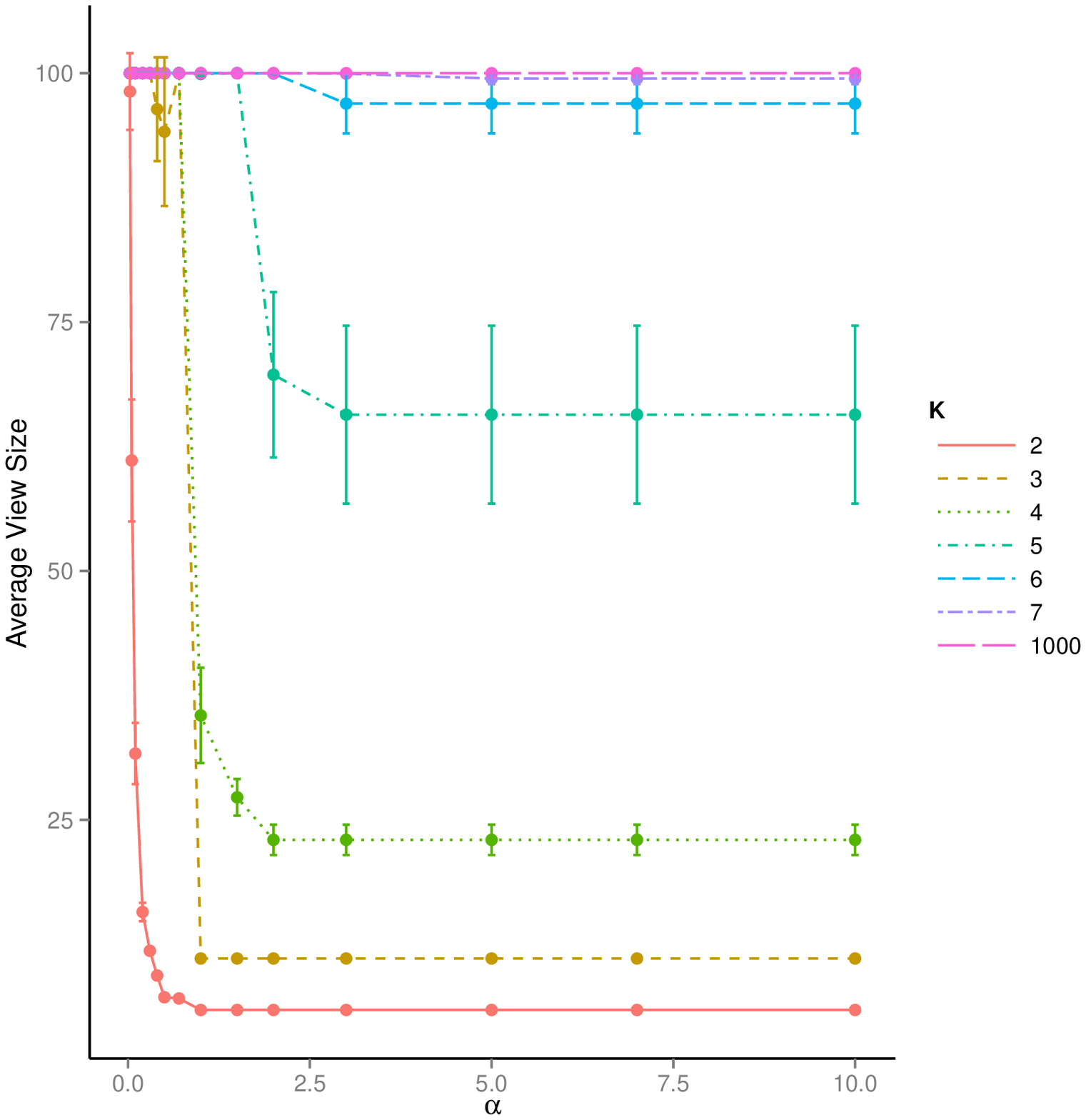}
	\includegraphics[width=0.5\textwidth]{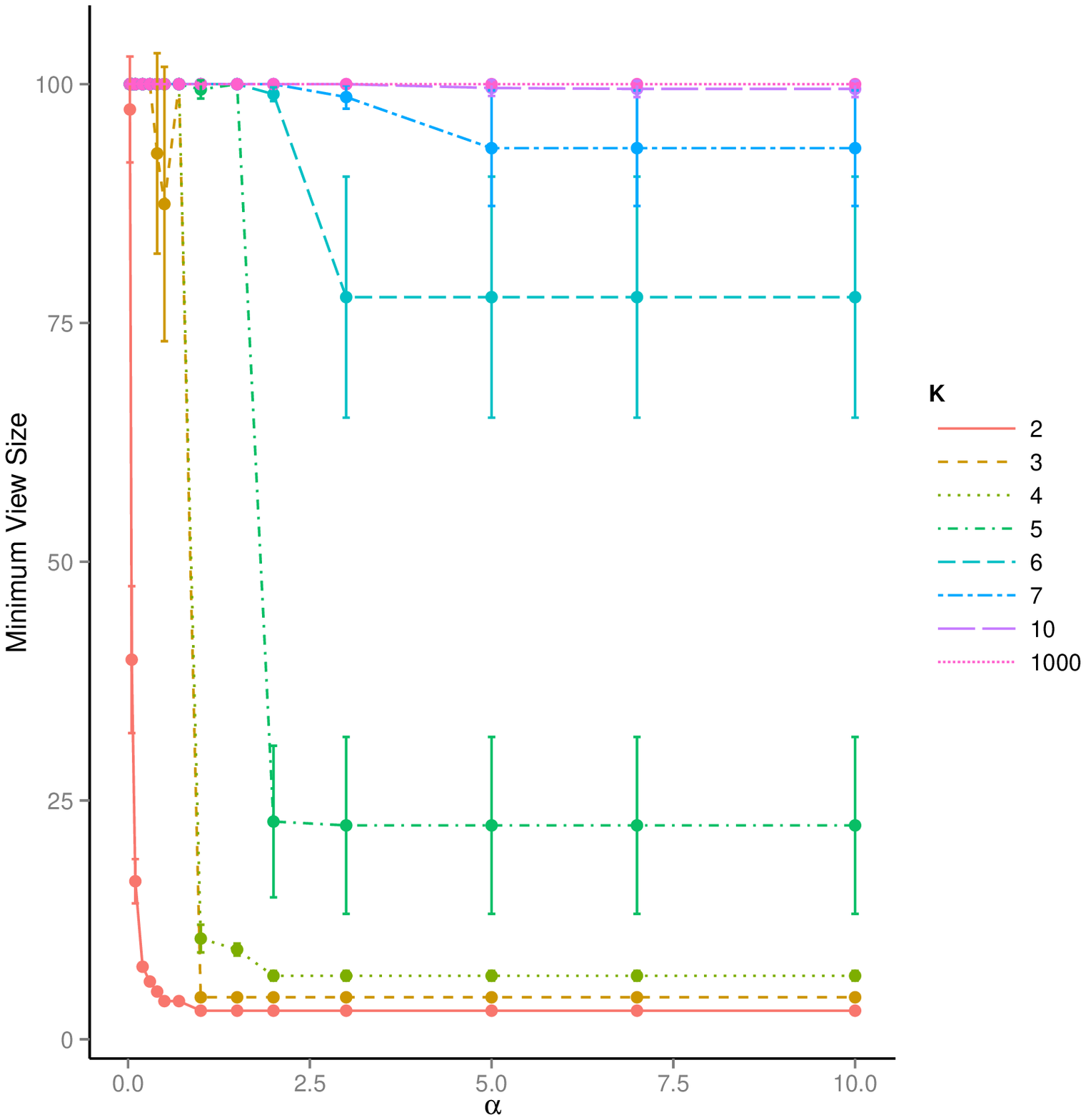}
	\caption{Minimum and average number of vertices in the players' view on stable networks as a function of $\alpha$ for the various values of $k$. Points correspond to mean values over $20$ different trees with $100$ vertices, while error bars represent $95\%$ confidence intervals. For the sake of readability, lines that coincide with the one for $k=1000$ are not depicted.}
	\label{fig:viewsize}
\end{figure}

Let us consider Figure~\ref{fig:viewsize}, where both the mean and the minimum number of vertices in the players' view is reported. The actual value of each point is the average over the different equilibria obtained from $20$ trees with $100$ vertices each (using the same set of parameters).
The same measurements on the other classes of graphs exhibit essentially the same behavior and are not reported.

Clearly, the view of a player decreases as $\alpha$ increases, and rapidly grows as $k$ becomes larger. Intuitively, the former is due to the fact that whenever $\alpha$ is small, players are more prone to buy a large number of edges. Regarding the latter, we observe that for the interesting threshold value of $k=7$, we have that, on average, a player knows more than $99$ vertices. Moreover, the minimum size is also pretty high, since even the player who knows the least portion of the network is able to see more than $93$ vertices (on average).

\paragraph{Quality of equilibria}

We now discuss how our theoretical bounds on the PoA compare to the experimental evaluation of the quality of equilibria, i.e., the ratio between the social cost of the attained equilibrium and the social optimum.

\begin{figure}[t]
	\includegraphics[width=0.5\textwidth]{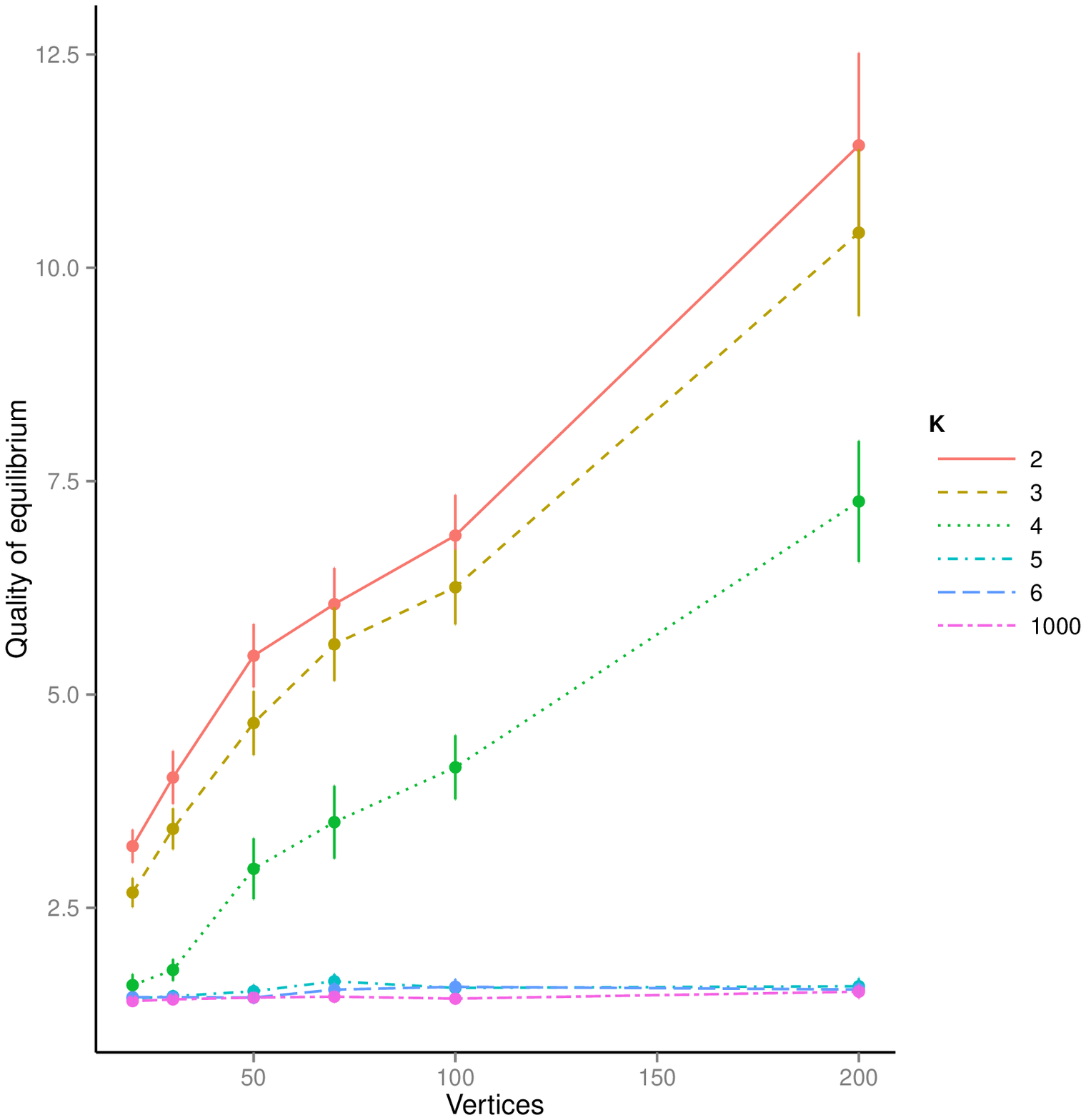}
	\includegraphics[width=0.5\textwidth]{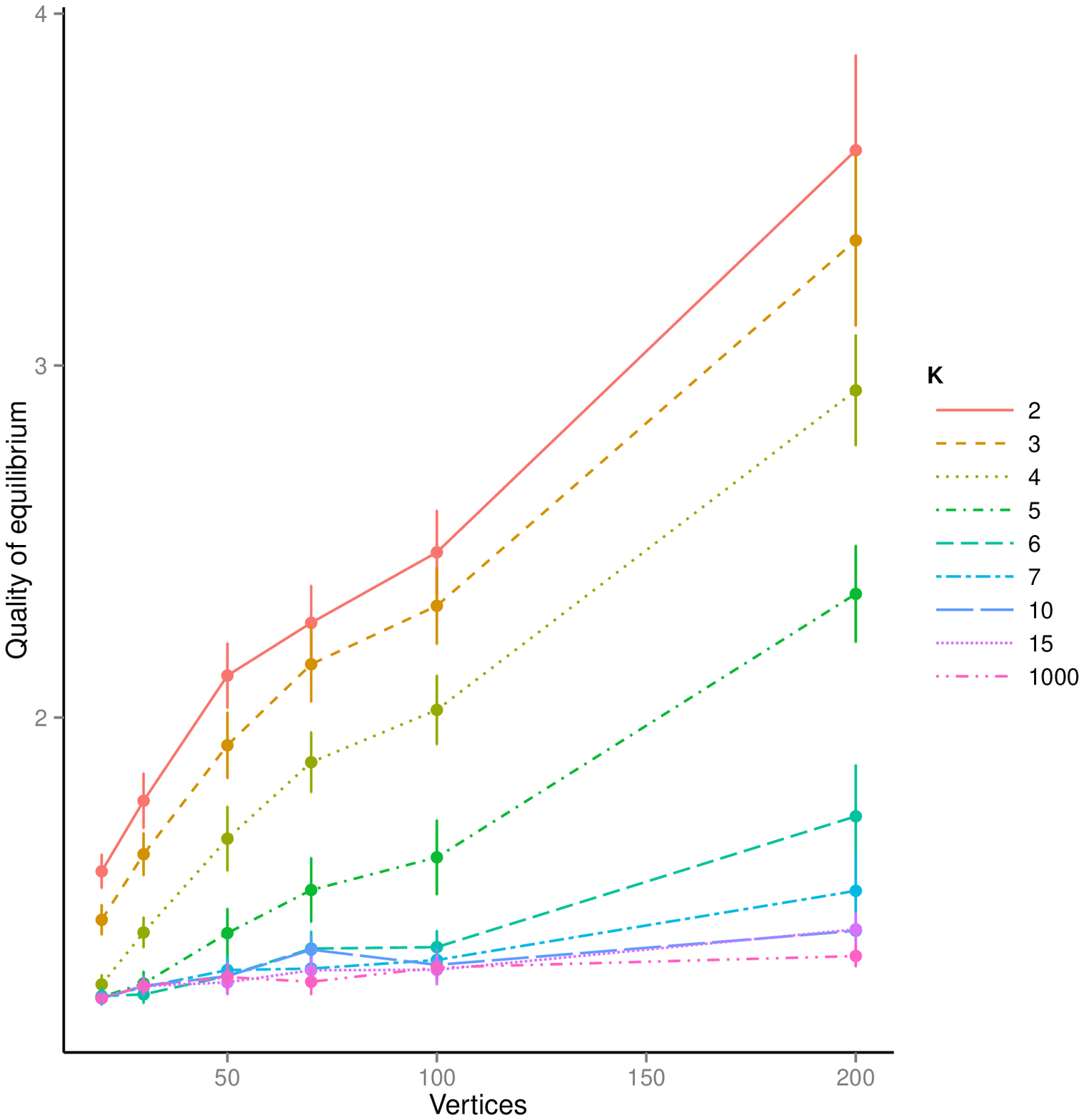}
	\caption{Quality of the stable networks as a function of $n$ for various values of $k$ on random trees. The left picture refers to $\alpha=1$ while the right one refers to $\alpha=10$. Points correspond to mean values over $20$ different trees, while error bars represent $95\%$ confidence intervals. For the sake of readability, lines that essentially coincide with the one for $k=1000$ are not depicted. Notice that for small values of $k$ the quality degrades linearly while for larger values of $k$ the PoA is almost constant.}
	\label{fig:poa}
\end{figure}

First of all, notice that our bounds show that the PoA decreases as $\alpha$ and/or $k$ increase. Moreover, for fixed values of the parameters, the PoA is $\Theta(n)$ whenever $\alpha \ge 2$.
This trend can be easily recognized in Figure~\ref{fig:poa} (b), where the quality of the equilibria on experimented trees is plotted against the number of vertices for $\alpha=10$. Notice that as soon as $k$ exceeds $6$, the quality of equilibria drastically improves. This is not surprising since, as already discussed, the players have a (almost) full-knowledge of the stable networks, and hence the (almost) constant bounds to the PoA given for the classical version of \MaxNCG hold.

Let us now focus on smaller values of $\alpha$, for example $\alpha=1$ as shown in Figure~\ref{fig:poa} (a). The general trend of the quality of equilibria is similar, except for the fact that $k \ge 5$ is now sufficient to yield full-knowledge equilibria. By contrast, our theoretical upper bound for smaller values of $k$ of
$O\bigg( n^{\frac{2}{\alpha}} + \min \{\frac{n\alpha}{k^2},  \frac{nk}{\alpha 2^{\Theta(\log^2 \frac{k}{\alpha})}} \} \bigg)$ is now $O(n^{\frac{2}{\alpha}} + n) = O(n^2)$.
This comes from the first term of the sum, which is related to the density of the equilibria. In our experiments, however, the resulting stable networks happened to be sparse, so that the usage cost prevails instead. This suggests that our density bounds might be too coarse, and this is somewhat supported by the corresponding lower bound of $\Omega\bigg( \max\{ \frac{n}{\alpha 2^{\Theta(\log^2 \frac{k}{\alpha})}}, n^\frac{1}{2k-2} \}\bigg)$, which now simplifies to $\Omega(n)$ whenever the density term is dominated by the usage cost term.

\begin{figure}[t]
	\includegraphics[width=0.5\textwidth]{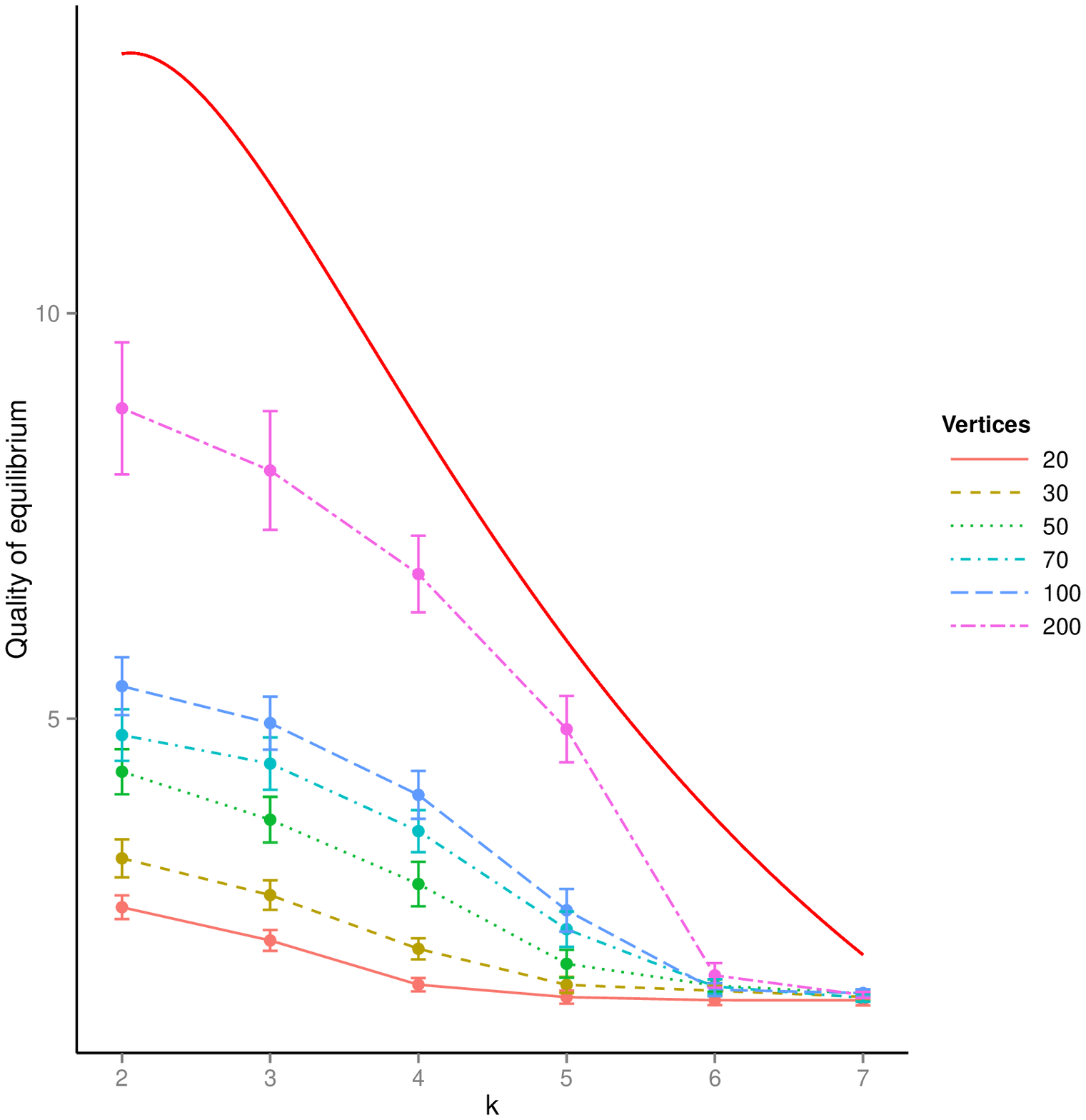}
	\includegraphics[width=0.5\textwidth]{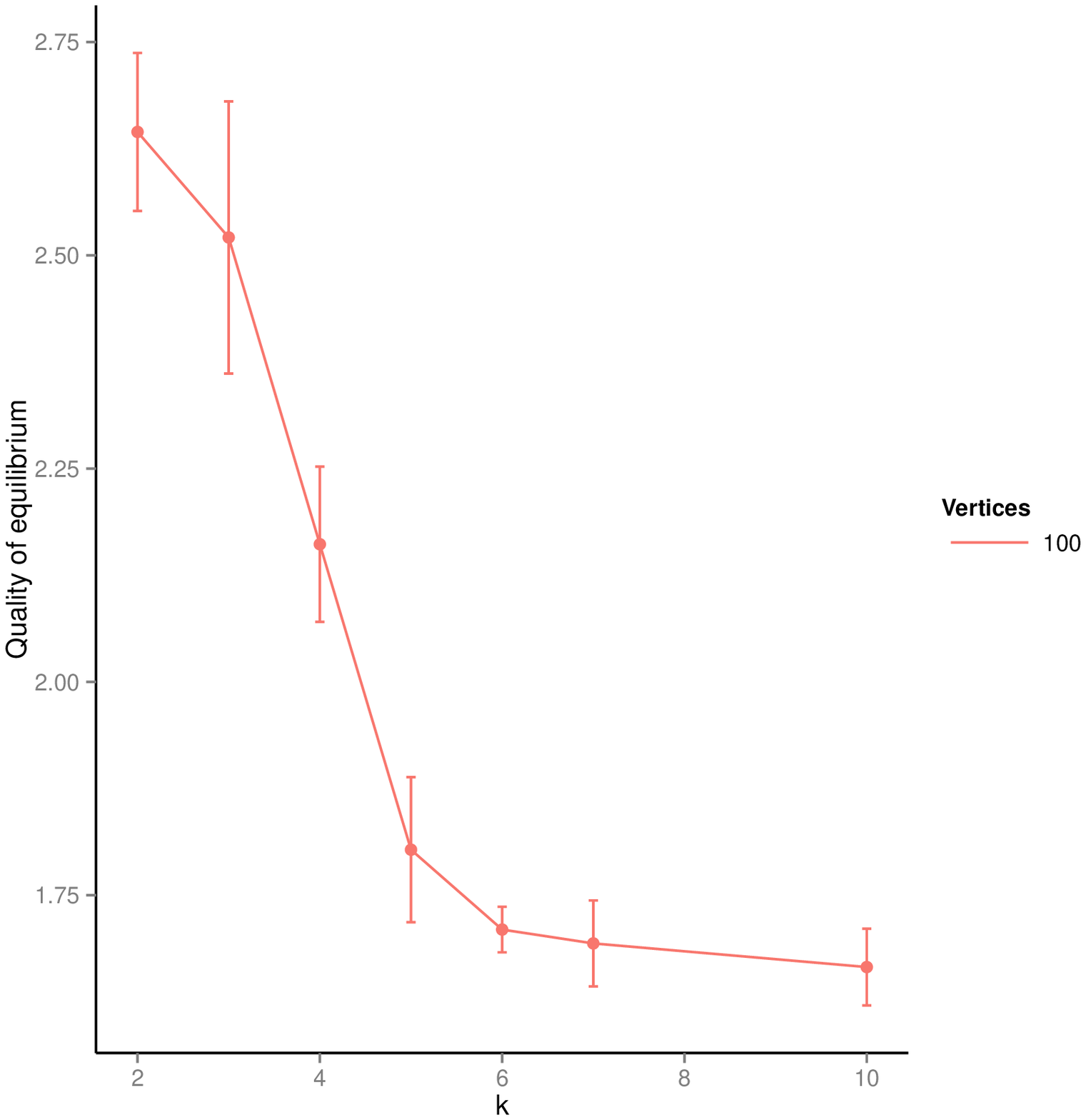}
	\caption{Quality of the stable networks as a function of $k$ for various values of $n$, and $\alpha=2$. The bold red line represents the trend of our theoretical upper bound. The left picture refers to random trees while the right one refers to random graphs with $n=100$ and $p=0.2$. Points correspond to mean values over $20$ different graphs, while error bars represent $95\%$ confidence intervals.}
	\label{fig:poa-k}
\end{figure}

We now examine in more detail how the quality of equilibria decreases as a function of $k$. Notice that, our theoretical upper bound to the PoA of $O\bigg( n^{\frac{2}{\alpha}} + \min \{\frac{n\alpha}{k^2},  \frac{nk}{\alpha 2^{\Theta(\log^2 \frac{k}{\alpha})}} \} \bigg)$, reduces to $f(k)=O(\frac{k}{2^{\log^2 k}})$, once $\alpha \ge 2$ and $n$ are fixed to be constant. In Figure~\ref{fig:poa-k} we report the measured quality of equilibria as a function of $k$ on both random trees (left) and graphs (right), for different values of $n$, and $\alpha=2$. Moreover, the trend of $f(k)$ is shown in red as a benchmark. For different values of $\alpha$ the trend is very similar, but the actual quality of equilibria scales down as expected.

\paragraph{Fairness of equilibria}

\begin{figure}[t]
	\includegraphics[width=0.5\textwidth]{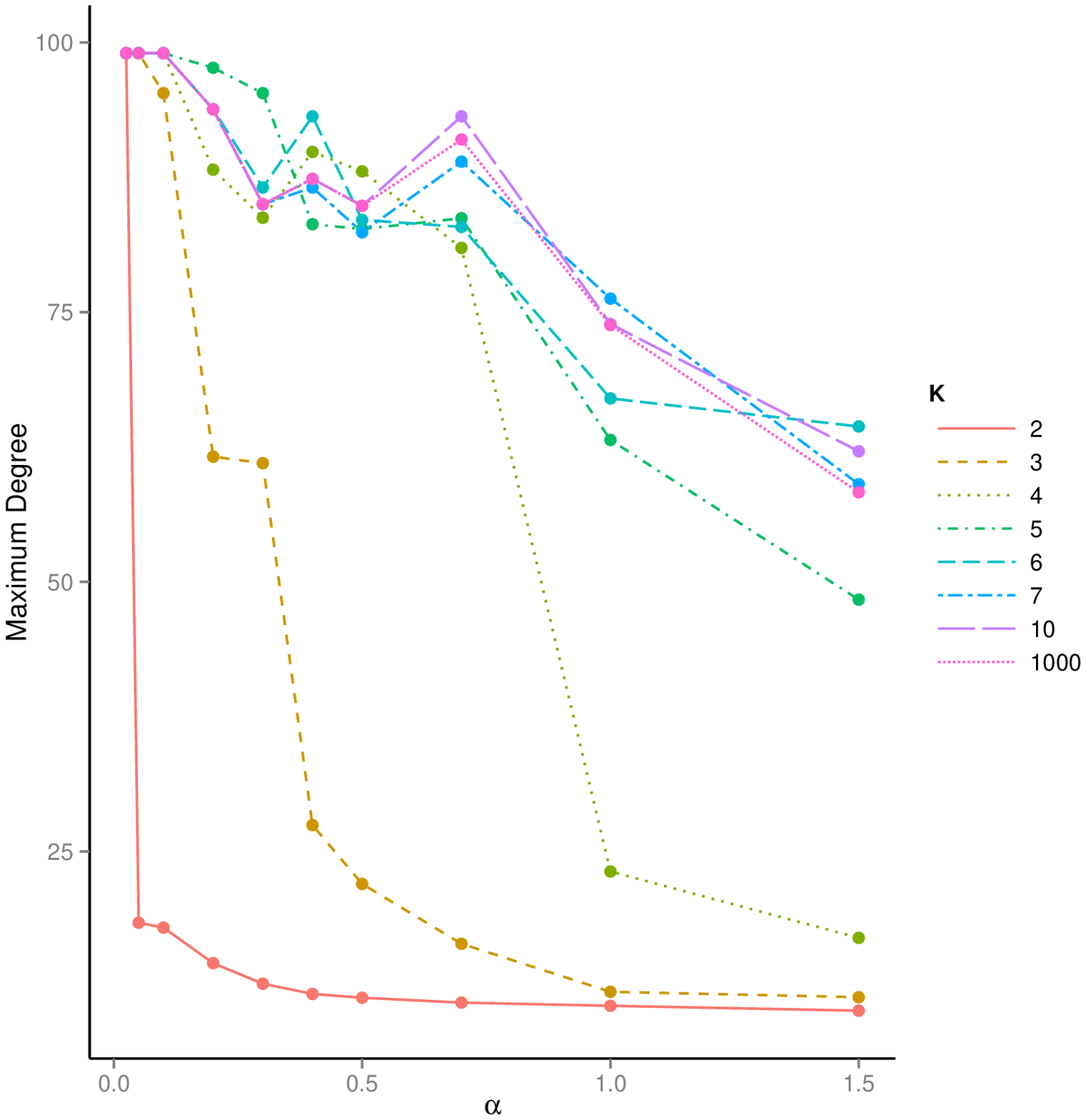}
	\includegraphics[width=0.5\textwidth]{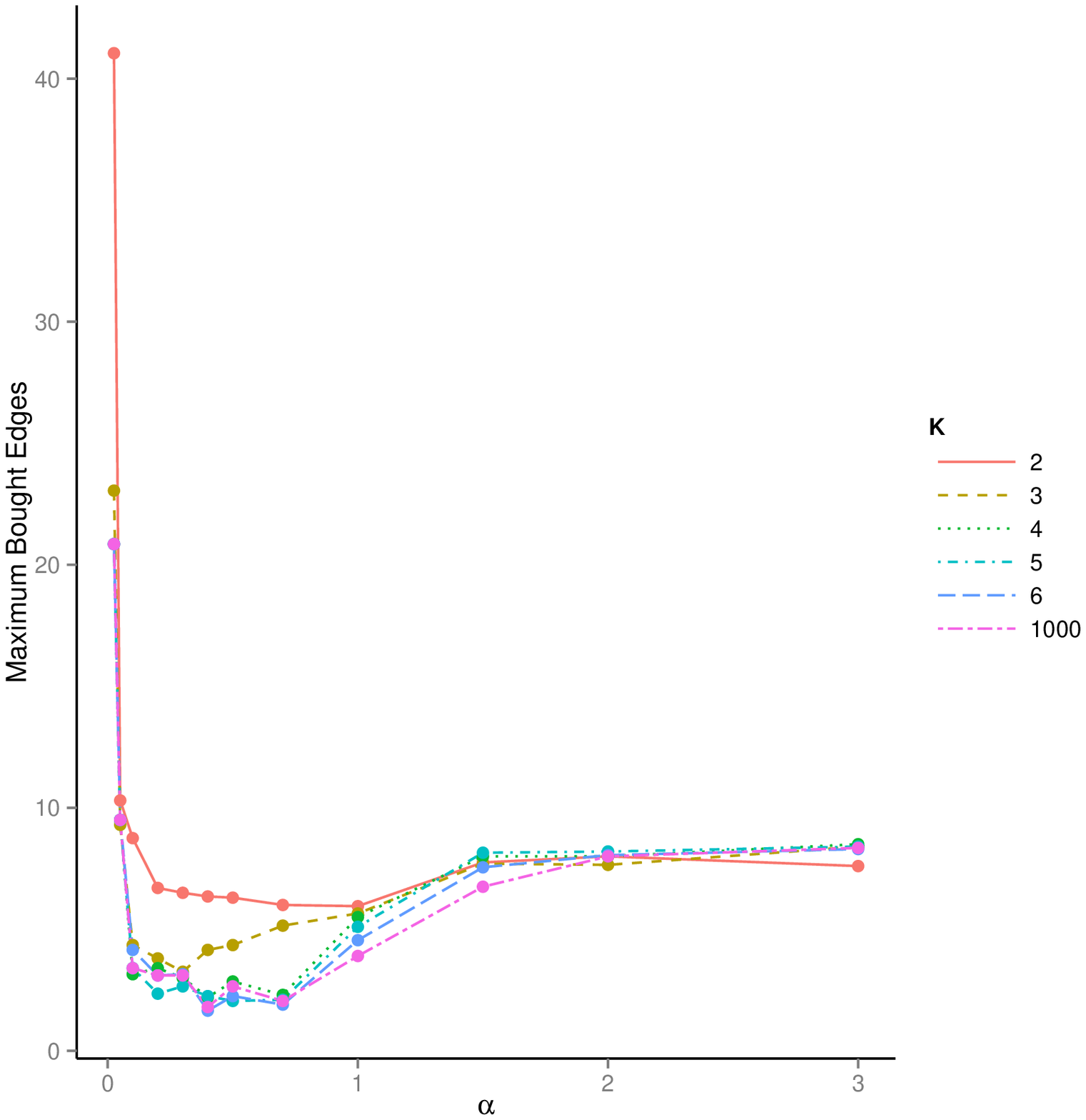}
	\caption{Maximum degree (left) and maximum number of bought edges (right) as a function of $\alpha$ for various values of $k$. Points correspond to mean values over $20$ different random graphs with $100$ vertices and $p=0.1$. For the sake of readability, lines that essentially coincide with the one for $k=1000$ are not depicted.}
	\label{fig:maxdeg}
\end{figure}

\begin{figure}[t]
	\centering
	\includegraphics[width=0.5\textwidth]{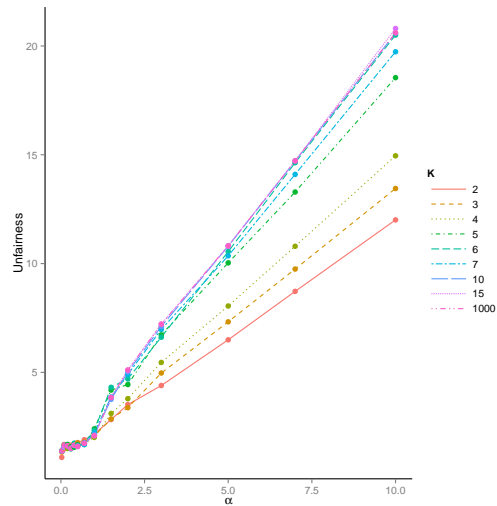}
	\caption{Unfairness ratio, i.e., the ratio between the highest and the lowest of players' costs, as a function of $\alpha$ for various values of $k$. Points correspond to mean values over $20$ different random graphs with $100$ vertices and $p=0.1$. For the sake of readability, lines that coincide with the one for $k=1000$ are not depicted. Notice small values of $k$ yield more fair equilibria. }
	\label{fig:fairness}
\end{figure}

We also collected statistics about the maximum degree of stable networks and the maximum number of edges bought by the players (see Figure~\ref{fig:maxdeg}, concerned with random graphs of $100$ nodes and with $p=0.1$). It is worth to notice that for $k \ge 4$ and small values of $\alpha$, we have that the maximum degree is more than $80$, while a player does not buy more than $9$ edges (for every values of $\alpha$ and $k$). This has a nice consequence in terms of fairness of stable networks, as shown in Figure~\ref{fig:fairness}, which decreases as $k$ increases. This would suggest that restricting the view of the players  could help to converge towards stable networks where players' costs do not differ too much.

\paragraph{Convergence time}

\begin{figure}[t]
	\includegraphics[width=0.5\textwidth]{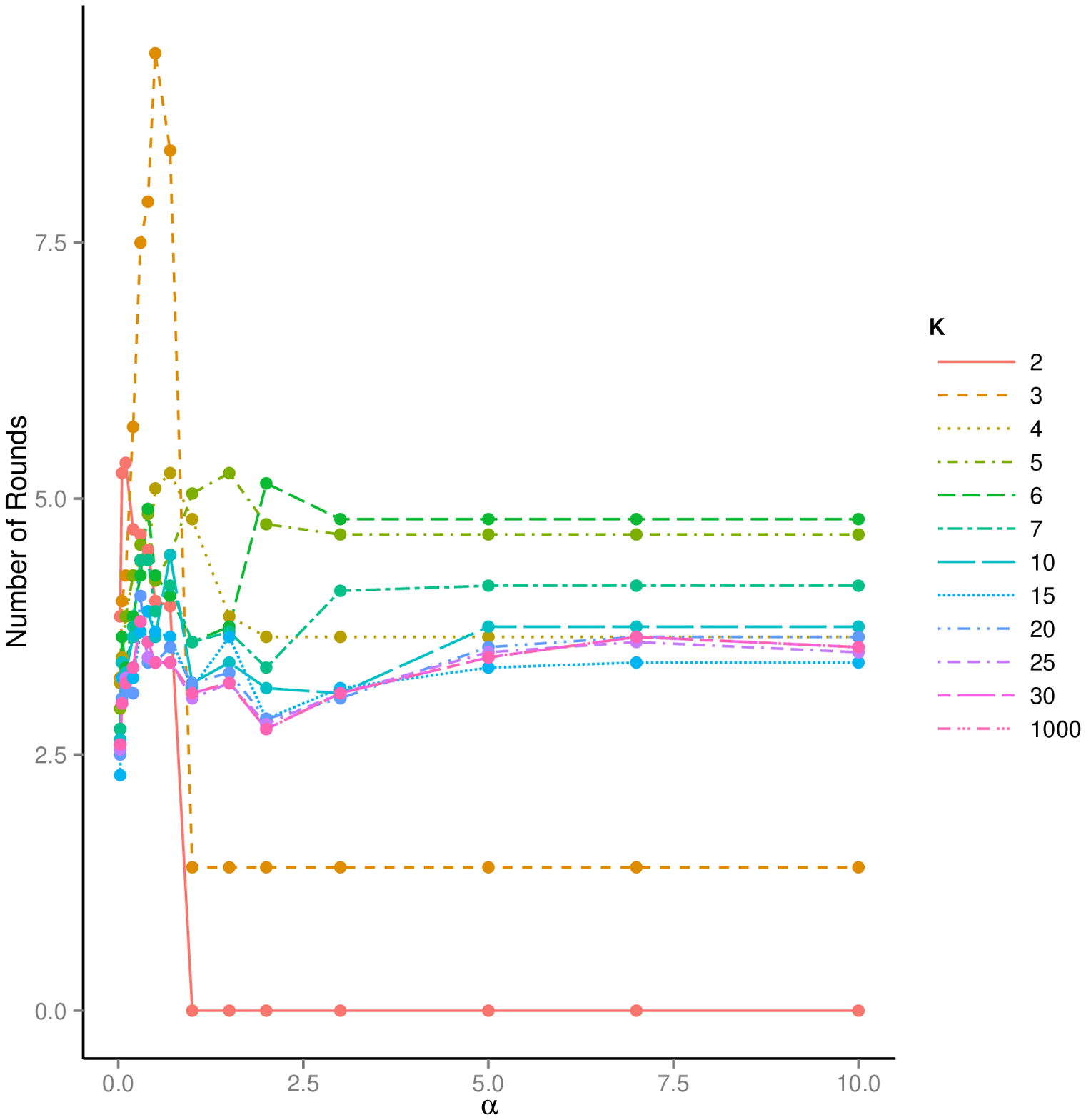}
	\includegraphics[width=0.5\textwidth]{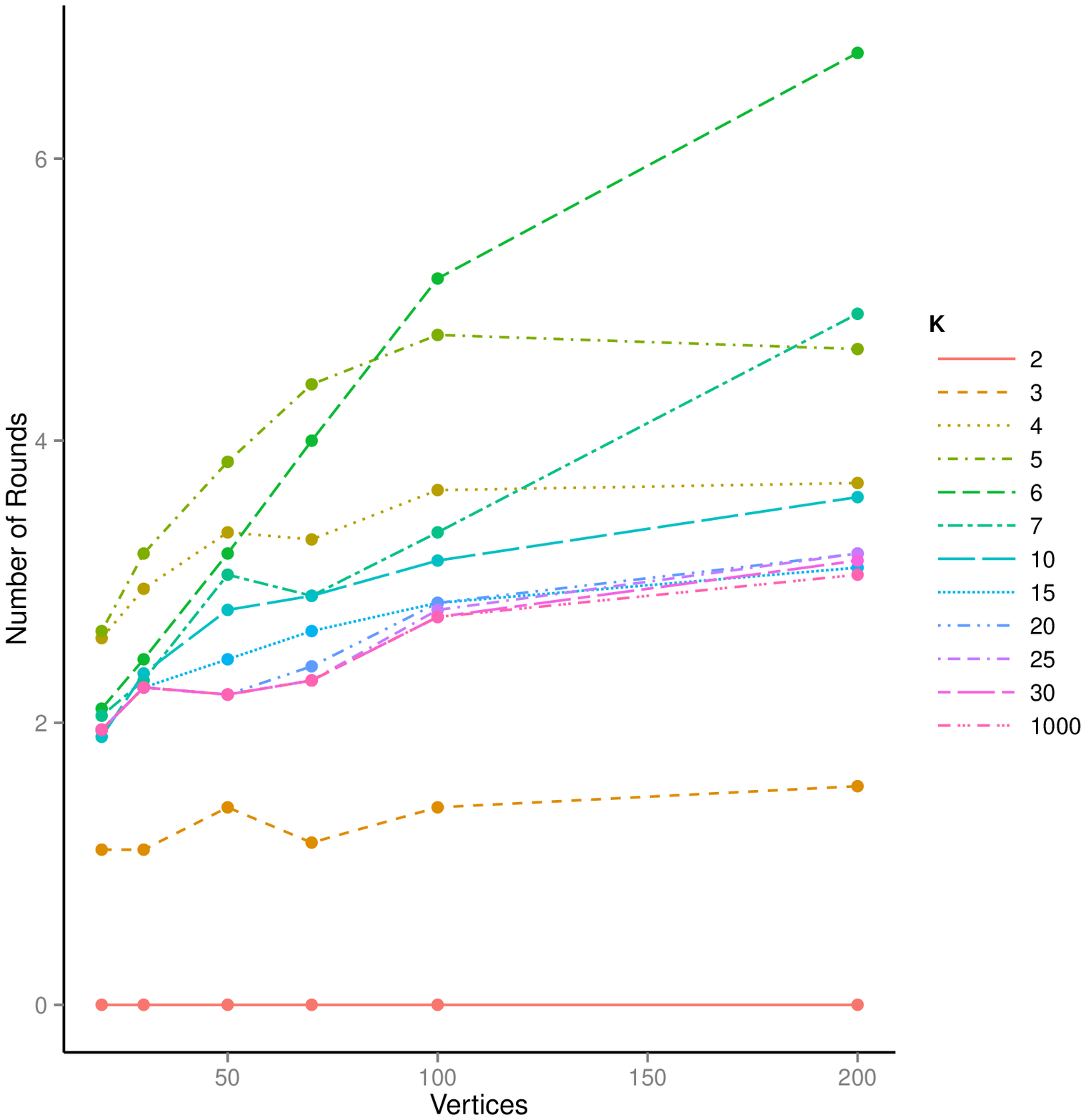}
	\caption{Number of rounds needed to converge to a stable network as a function of $\alpha$ with $n=100$ (left) and as a function of $n$ with $\alpha=2$ (right) for various values of $k$. Points correspond to mean values over $20$ different random trees.}
	\label{fig:convergence}
\end{figure}

To conclude this section let us focus on the convergence time, i.e., the number of rounds needed to actually reach an equilibrium (if any). As we pointed out in the technical part, the convergence of a best-response dynamics is not guaranteed. In practice, however, it seems very common: we simulated about $36\,000$ best-response dynamics, and only encountered best-response cycles in $5$ of them.
In all the other cases convergence appeared to be pretty fast, as shown in Figure~\ref{fig:convergence} for random trees. In fact, starting from both random trees and Erd\H{o}s-Rényi graphs, for almost every combination of $\alpha$ and $k$, in more than $95\%$ of the times, at most $7$ rounds are enough to converge to a stable network. Finally, the total number of rounds increases with $n$ as one would expect, although this happens quite slowly for almost every choice of the parameters.

\section{Conclusions}
\label{sec:conclusions}

In this paper we have studied the game-theoretic and computational implications of a limited players' view (more precisely, confined to a prescribed distance radius from each player) in the autonomous formation of a (large) network. In this scenario, we developed a systematic analysis on the PoA of the two classic variants of the game, namely \textsc{MaxNCG} and \textsc{SumNCG}, along with an extensive experimental study, which, for the computational feasibility reasons that have been explained in the paper, were limited to the former game, though.

Concerning our future research, on one hand we plan to investigate some of the issues that we left (partially) open, in particular the PoA space for the \SumNCG needs to be explored in more detail. On the other hand, we believe that our incomplete-knowledge approach deserves to be extended in several directions. As a first step in this regard, in \cite{BGLP14} we have considered three local-knowledge models usually adopted in {\sc Network Discovery}, i.e., the optimization problem of reconstructing the topology of an unknown network through a minimum number of queries at its nodes. For these models, we provided exhaustive answers to the canonical algorithmic game theoretic questions w.r.t. our LKE concept.
Besides studying new models, however, we feel that our own model has still few intriguing issues that should be investigated. Above all, it would be interesting to relax our worst-case approach, and analyze a NCG under a Bayesian perspective. Finally, we feel that the local-knowledge model could be extended to other (network) game-theoretic settings, given that the global knowledge is a recurring, still critical, assumption.

\section*{Acknowledgments}
We wish to thank Stefano Smriglio for useful discussions concerning the experimental section of the paper.

% Bibliography
\bibliographystyle{ACM-Reference-Format-Journals}
\bibliography{bibliography}

\received{Month Year}{Month Year}{Month Year}

\end{document}